\documentclass[12pt,twoside,english]{extarticle}
\usepackage{mathpazo}

\usepackage[T1]{fontenc}
\usepackage[latin9]{inputenc}
\usepackage[letterpaper]{geometry}
\usepackage[active]{srcltx}
\usepackage{babel}
\usepackage{array}
\usepackage{booktabs}
\usepackage{textcomp}
\usepackage{mathrsfs}
\usepackage{enumitem}
\usepackage{amsmath}
\usepackage{amsthm}
\usepackage{amssymb}
\usepackage{setspace}
\usepackage[authoryear]{natbib}
\PassOptionsToPackage{normalem}{ulem}
\usepackage{ulem}
\usepackage[unicode=true,
 bookmarks=true,bookmarksnumbered=false,bookmarksopen=false,
 breaklinks=false,pdfborder={0 0 0},pdfborderstyle={},backref=false,colorlinks=false]
 {hyperref}

\makeatletter

\providecommand{\tabularnewline}{\\}

\@ifundefined{lettrine}{\usepackage{lettrine}}{}
\theoremstyle{definition}
\newtheorem{defn}{\protect\definitionname}
\theoremstyle{plain}
\newtheorem{prop}{\protect\propositionname}
\theoremstyle{plain}
\newtheorem{cor}{\protect\corollaryname}
\theoremstyle{plain}
\newtheorem{assumption}{\protect\assumptionname}
\theoremstyle{plain}
\newtheorem{thm}{\protect\theoremname}

\date{}
\usepackage{tikz}
\usepackage{pgfplots}
\usepackage{relsize}
\usepackage{varwidth}
\usepackage{enumitem}
\usetikzlibrary{backgrounds}
\usepackage[multiple,bottom]{footmisc}
\usepackage{varioref}
\addto\extrasenglish{
}
\usepackage{amsmath}
\usepackage{amsthm}
\usepackage{thmtools}
\usepackage{bbm}
\allowdisplaybreaks

\DeclareMathOperator*{\marg}{marg}
\DeclareMathOperator*{\supp}{supp}

\usepackage{titlesec}
\usetikzlibrary{shapes}
\usetikzlibrary{plotmarks}
\usetikzlibrary{patterns}
\tikzset{
  mycross/.pic={
    \draw[pic actions] 
      (-4pt,0) -- (4pt,0)
      (0,-4pt) -- (0,4pt);
  },
}

\makeatother

\providecommand{\assumptionname}{Assumption}
\providecommand{\corollaryname}{Corollary}
\providecommand{\definitionname}{Definition}
\providecommand{\propositionname}{Proposition}
\providecommand{\theoremname}{Theorem}

\begin{document}
\title{Arbitrage from a Bayesian's Perspective\thanks{Version: May 25, 2022.}}
\author{Ayan Bhattacharya\thanks{Email: ayan.bhattacharya@gmail.com. Affiliation: University of Chicago
\textendash{} Booth School of Business, and Arrow Markets. My thanks
to Stefan Nagel for helpful comments.}}

\maketitle
\begin{singlespace}
\begin{abstract}
This paper builds a model of interactive belief hierarchies to derive
the conditions under which judging an arbitrage opportunity requires
Bayesian market participants to exercise their higher order beliefs.
Methodologically, the approach of the paper is to take the standard
asset pricing setup that gives rise to arbitrage and transform it
into a Bayesian decision problem faced by a representative market
agent. As a Bayesian, such an agent must carry a complete recursion
of priors over the uncertainty about future asset payouts, the strategies
employed by other market participants that are aggregated in the price,
other market participants' beliefs about the agent's strategy, other
market participants beliefs about what the agent believes their strategies
to be, and so on ad infinitum. Defining this infinite recursion of
priors \textemdash{} the belief hierarchy so to speak \textemdash{}
along with how they update gives the Bayesian decision problem equivalent
to the standard asset pricing formulation of the question. In this
setting, any update to the belief hierarchy of an agent is one of
two kinds: a change in belief about the asset payouts, or a change
in belief about the strategies and beliefs employed by other market
participants. The main results of the paper show that an arbitrage
trade corresponds to special updates of the second kind. When an agent
anticipates market participant responses will be generated using $k$
levels of the belief hierarchy but finds that the actual asset prices
are supported by $k+1$ or higher levels, there is an arbitrage opportunity.
It is shown that the presence of arbitrage depends on the degree of
optimality of the belief hierarchies employed by market agents and
responsiveness of the price aggregation mechanism, and is closely
related to market tatonnement. The paper connects the foundations
of finance to the foundations of game theory by identifying a bridge
from market arbitrage to market participant belief hierarchies.\bigskip{}

{\footnotesize{}\noindent Keywords: Arbitrage opportunity, Beliefs
in asset price, Level-k reasoning, Higher-order belief, Bayesian belief
hierarchy, Fundamental theorem of asset pricing.}{\footnotesize\par}
\end{abstract}
\end{singlespace}

\thispagestyle{empty}

\newpage{}

\section{Introduction}

This paper connects the foundations of finance to the foundations
of game theory. Specifically, it builds a model of interactive belief
hierarchies to derive the theoretical conditions under which judging
an arbitrage opportunity requires Bayesian market participants to
exercise their higher order beliefs. A no-arbitrage argument is at
the heart of many of the modern developments in asset pricing theory
and \citet{key-18} term the equivalence results linking the various
characterizations of no-arbitrage the ``fundamental theorem of asset
pricing.'' Likewise, the ``hierarchy of beliefs'' technique proposed
in \citet{key-1} for Bayesian players, that evolved into the epistemic
approach to the field, underpins many of the important developments
in game theory. A rigorous bridge from market arbitrage to market
participant belief hierarchies, therefore, offers a potential path
to transferring strategic concerns of agents that underlie game theoric
modeling to asset pricing theory. 

Methodologically, the approach of the paper is to take the standard
asset pricing setup that gives rise to arbitrage and transform it
into a Bayesian decision problem faced by a representative market
agent. The traditional definition of an arbitrage portfolio is given
in terms of a set of conditions on the ex-ante prices of assets relative
to their ex-post payouts. Nonetheless, the ex-ante prices arise as
a result of individual optimization of some sort or another, and market
aggregation, so the price formation process embeds an implicit strategic
interaction among the participants. When taking this interaction into
account, the uncertainty facing an agent in the market is no longer
just about the ex-post payouts from assets, but also about the strategies
chosen by other market participants that get aggregated into the price.
An agent must carry a prior over this uncertainty as a Bayesian. That
is not all, though. The agent is also uncertain about what other market
participants believe her strategy to be, about what other market participants
believe about what she believes their strategies to be, and so on
ad infinitum. As a Bayesian, an agent must carry a prior over each
of these layers of uncertainty. Defining this infinite recursion of
priors \textemdash{} the belief hierarchy so to speak \textemdash{}
along with how they update gives the Bayesian decision problem equivalent
to the standard asset pricing formulation of the question.

A fundamental insight in epistemic game theory is that if the basic
uncertainty spaces are suitably restricted, one may define a universal
domain of uncertainty for an agent that encompasses all the recursive
layers of uncertainty arising in an interaction. Defining a prior
over this universal domain is then equivalent to defining the infinite
recursion of priors, as long as the agents' beliefs obey certain minimal
consistency requirements. We adopt this epistemic approach in the
paper. Any update to the belief hierarchy of a market agent could
then be one of two kinds: a change in belief about the asset payoffs,
or a change in belief about the strategies and beliefs employed by
other market participants. In this paper we switch off the first type
of update by assuming symmetric information about ex-post payouts
and focus exclusively on updates of the the second kind. That is to
say, an agent starts with a prior belief hierarchy about the market
environment she is facing, and after seeing the actual asset prices
generated as a result of the choices made by market participants,
revises her belief hierarchy if she has to. The main results of the
paper show that an arbitrage trade corresponds to a subset of revisions
of this type.

An arbitrage opportunity is not an equilibrium phenomenon, so the
market aggregation mapping that gives rise to such an opportunity
cannot be the standard individual-optimization/market-clearing pricing
function that is used to study equilibrium. Nevertheless, if the mapping
that aggregates individual choices into a market price is known to
a market participant, she may in-principle back out a set of choices
for other market participants that are consistent with the market
price she observes. In other words, there is a set of market participant
belief hierarchies that support an asset price which can be backed
out by a market agent who observes the asset price and knows the aggregation
mapping. An arbitrage opportunity arises when an agent finds that
the belief hierarchies in actual use are in some sense ``more optimal''
than what she had anticipated them to be. 

Optimality in our setup is defined in terms of the number of levels
of the belief hierarchy that are taken into account by agents to generate
undominated responses \textemdash{} responses that cannot be unequivocally
improved. When an agent anticipates that other market participants
will be generating their responses using only $k$ levels of their
belief hierarchy but subsequently discovers \textemdash{} when backing
out actual choices from the price \textemdash{} that they are using
$k+1$ or higher levels, there is an arbitrage opportunity. For example,
the agent might anticipate that the rest of the market will use undominated
responses that account for two levels of their respective hierarchies,
and choose her initial strategy to be an optimal response to this
anticipation. For this agent to discover an arbitrage opportunity
in the market subsequently, it is necessary that the rest of the market
actually use undominated responses that account for at least three
levels of their hierarchy.

Whether or not an arbitrage trade exists depends both on the degree
of optimality of the anticipated belief hierarchies and on the extent
to which the market aggregation mapping allows an agent to uncover
the actual belief hierarchies supporting an observed price. It is
the interplay between these two factors that gives the study of arbitrage
with belief hierarchies its distinct flavor. For example, if the market
aggregation is a constant mapping (i.e., prices stay unchanged no
matter what choices market participants make) there is no arbitrage
opportunity no matter what belief hierarchy agents employ. This is
because an agent simply cannot discern any feature of the actual belief
hierarchies that are supporting the observed prices. So, no matter
what belief hierarchy she anticipates, after observing the prices
the agent is unable to make the case that she would do unequivocally
better if she used a different hierarchy. On the other hand, if the
market aggregation mapping is one-to-one (i.e., observed prices allow
agents to perfectly distinguish the belief hierarchies actually used)
then no-arbitrage implies that agents are optimizing over the entire
infinite hierarchy. This is because a competitive game ensues among
the market participants: every agent wants to optimize one level more
than the rest of the market so as to not leave any money on the table,
in effect pushing the orders of optimization to infinity for everyone.
Such a competitive process would not go all the way to infinity if
an agent could not distinguish her counterparty's belief hierarchy
with sufficient precision from the prices.

The analysis in the paper highlights the importance of higher order
reasoning in markets. It is not enough if an agent reasons optimally
about just fundamentals or beliefs of other market participants; to
leave no money on the table, the agent must also reason optimally
about beliefs about beliefs of other market participants, about beliefs
about beliefs about beliefs of other market participants, and so on,
as far as needed under the market aggregation mapping. Yet common
experience seems to suggest that traders in the real world don't go
very far up the orders when reasoning deliberately about their environment.
The resolution to this apparent incongruity lies in recognizing that
a market tatonnement process \textemdash{} in which traders are simply
responding to immediate market circumstances \textemdash{} can deliver
the same outcome as a higher order reasoning process. At each step
a trader could be reasoning deliberately just one-step ahead, yet
stacking up a series of such ``one-step aheads'' in a tatonnement
sequence leads to market outcomes that are indistinguishable from
higher order reasoning.

\medskip

\noindent\textbf{Related Literature:} The need to study how agents
forecast the forecast of others has motivated financial economics
ever since \citet{key-29} introduced his influential metaphor of
markets as a beauty contest. Beginning with the seminal work of Townsend
\citeyearpar{key-31,key-30}, \citet{key-32} and \citet{key-33},
there is a large literature at the intersection of finance and macroeconomics
that looks at models where agents have heterogeneous expectations
about the future realizations of economic variables, and thus have
to forecast the forecast of others. In asset pricing theory, papers
like \citet{key-37}, \citet{key-34}, \citet{key-36} and \citet{key-35}
build on this metaphor to create rational (or near rational) expectations
equilibrium models of asset prices that rely on heterogeneous beliefs
and differential information among agents. While related to this literature,
the main point of departure of the present study is that neither heterogeneous
expectations about fundamentals nor asymmetric information play any
role in the model. In fact, it is assumed throughout that the physical
probability measure describing the ex-post payoff uncertainty is known
to all the market participants.

On the epistemic side, a number of papers have explored both the belief-about-belief
based foundation as well as reasoning based foundation of rational
expectations equilibrium in economies with asymmetric information.
In the reasoning based ``eductive'' approach pioneered by Guesnerie
(see Guesnerie \citeyear{key-40}, \citeyear{key-38} for important
papers; and \citealt{key-39} and \citealt{key-41} for surveys of
this literature) the equilibrium solutions are based on the assumption
that each agent reasons about the reasoning of other agents in line
with common knowledge of rationality and the model. To microfound
the adaptive learning that is necessary in reasoning based models,
\citet{key-210} provide a decision theoretic framework based on agents
who are ``internally rational'' but may not be ``externally rational''.
On the other hand, the belief-about-belief based approach that hews
more closely to the traditional epistemic foundation for solution
concepts relies on various alternative notions of common knowledge
of market clearing, player beliefs and rationality to justify rational
expectation outcomes (see \citealt{key-45}, \citealt{key-43}, \citealt{key-42},
\citealt{key-44}). At a broad level, the approach in the present
paper is similar to the one taken in this literature. However, unlike
this body of work, there is no asymmetric information in the present
study. Further, arbitrage is not an equilibrium phenomenon, so the
rational expectations equilibrium framework is of limited use in understanding
the genesis of arbitrage trade.

The rigorous study of arbitrage was initiated in Ross (\citeyear{key-46},
\citeyear{key-17}) and \citet{key-16}, and a number of alternative
formulations for such trade were proposed in the ensuing years. \citet{key-18}
termed the equivalence of the various alternative conditions that
led to (no) arbitrage the fundamental theorem of asset pricing. One
way to interpret the results in this paper is that they add to this
list of equivalent conditions \textemdash{} in this case, using the
Bayesian belief hierarchies of market agents. Another way to think
about the results is that they highlight the fundamental significance
of level-$k$ reasoning (\citealt{key-47}, \citealt{key-48}, \citealt{key-49})
and the recursion of priors for asset pricing and arbitrage theory.
More references to the literature are interspersed throughout the
text.

\medskip

\noindent\textbf{Organization of the paper: }The main results of
the paper are in Section \ref{sec:No-arbitrage-and-Higher-Beliefs}
and the preceding sections build the tools that are necessary to derive
them. Section \ref{sec:No-arbitrage} describes the classical no-arbitrage
condition that is the starting point of the study. Section \ref{sec:Transforming Arbitrage}
then gradually transforms the classical definition into a form that
is more explicit about the strategic concerns of a representative
Bayesian market agent, with Section \ref{subsec:Model of market}
providing the model of the market that such an agent uses and Section
\ref{subsec:Arbitrage-in-M} defining arbitrage in the model. Section
\ref{sec:Hierarchies} describes the construction of the recursive
hierarchy of beliefs that are the building blocks for asset pricing
under strategic uncertainty, with Section \ref{subsec:Canonical}
outlining the basic hierarchy that is used in epistemic game theory
and Section \ref{subsec:Dominated-Responses} defining the specific
hierarchy that we use in our analysis. Section \ref{sec:No-arbitrage-and-Higher-Beliefs}
contains the main results of the paper linking traditional descriptions
of arbitrage to belief hierarchy based descriptions. Section \ref{subsec:Link-Between-Arbitrage-Hierarchy}
contains the central theorems of the paper characterizing arbitrage
and no-arbitrage in terms of belief hierarchies, and emphasizes the
importance of responsiveness of the market aggregation mapping for
arbitrage. Section \ref{subsec:Tatonnement Reasoning} highlights
the close analogy between tatonnement and reasoning-based price adjustment
processes. Section \ref{sec:Discussion} contains a discussion on
ways to empirically measure the orders of belief used by market participants.

\section{A First Definition of Arbitrage\label{sec:No-arbitrage}}

The textbook definition of an arbitrage opportunity (for example,
\citealt{key-6}) is a portfolio that requires no investment to set
up, is guaranteed to not lose money, and generates positive income
almost surely. More specifically, consider a financial market with
a vector of $d$ assets. One of the assets in the market is risk-free,
in the sense that it always pays out a fixed amount ex-post. The payoffs
from the other assets are stochastic. To model the fundamental uncertainty
in the market, let us fix a measurable space $(S,\mathcal{S})$. The
state space $S$ is assumed to be compact and metric, and the elements
of this space are labeled \emph{states of nature}. The ex-post asset
payouts are given to be non-negative, bounded random variables\footnote{To facilitate easy identification, random variables in the paper are
distinguished by a tilde above the variable.} 
\begin{equation}
\tilde{x}=(\tilde{x}^{1},\dots,\tilde{x}^{d}).
\end{equation}

The assets are priced in the market ex-ante (i.e., before the uncertainty
is realized), and let us denote the market prices by 
\begin{equation}
\text{\ensuremath{q}}=(q^{1},\dots,q^{d})\in\mathbb{R}^{d}.\label{eq:original price vector}
\end{equation}

A portfolio is a vector 
\begin{equation}
\text{\ensuremath{\theta}}=(\theta^{1},\dots,\theta^{d})\in\mathbb{R}^{d},
\end{equation}
where $\theta_{i}$ represents the number of units of the $i^{th}$
asset in the collection. The ex-ante cost of creating the portfolio
$\theta$ is $\theta\cdot q$, and the portfolio pays off $\theta\cdot\tilde{x}$
ex-post depending on the realized state of nature $s\in S$.
\begin{defn}
\textsf{(Arbitrage opportunity)\label{def:Basic-Arb Opportunity}
A portfolio $\theta\in\mathbb{R}^{d}$ is called an arbitrage opportunity
with respect to a probability measure $\mathbf{P}$ on the measurable
space $(S,\mathcal{S})$ if
\begin{equation}
\theta\cdot q\leq0,\,\,\text{but}\,\,\theta.\tilde{x}\geq0\,\mathbf{P}\text{-a.s.}\,\,\text{and}\,\,\mathbf{P}[\theta.\tilde{x}>0]>0.\label{eq:Arb defn-1}
\end{equation}
\smallskip}

Notice that the probability measure $\mathbf{P}$ is relevant in so
far as it fixes the null sets of the measurable space. One could alternatively
formulate the definition using the state space directly (a portfolio
that loses money under no event and makes money under some events)
but such a definition would come with the assumption that at least
some events that make money for the portfolio have a non-zero probability.

Trading models based on no-arbitrage usually take $\tilde{x}$ and
$q$ as model primitives and impose conditions that rule out arbitrage
opportunities in the market (for example, a positive stochastic discount
factor, or a risk-neutral measure condition). Our journey in the sequel
will be in the opposite direction. We start with the no-arbitrage
condition and try to obtain ``deeper'' foundations for the concept
in terms of the Bayesian beliefs of a market participant. 
\end{defn}

\section{Transforming the Arbitrage Definition \label{sec:Transforming Arbitrage}}

In this section, we set about transforming the definition of arbitrage
provided in the previous section to arrive at an interpretation that
is more amenable to strategic analysis (this interpretation in contained
in Proposition \ref{prop:Amenable-Arb-def}). We take the viewpoint
of a particular market participant \textemdash{} agent $i$ (she)
\textemdash{} who is reasoning about the market, and try to decipher
the meaning of arbitrage from her perspective (the terms \emph{market}
\emph{participant} and \emph{agent} are used interchangeably in the
sequel). The market under focus need not be in equilibrium, so the
analysis relies on deriving the value that our agent attaches to her
strategies given an arbitrary market stochastic discount factor (SDF)
\textemdash{} without explicitly specifying an equilibrium price formation
process. In a Bayesian universe, SDFs and strategies must themselves
arise as a result of beliefs held by agents, and in subsequent sections
we extend the analysis by connecting it explicitly to underlying belief
hierarchies.

The classical definition of arbitrage keeps silent on how the arbitrage
opportunity is actually exploited in the market. The implicit assumption
is that if prices satisfy the condition in Definition \ref{def:Basic-Arb Opportunity},
then at least one market participant will notice this and exploit
the profit-making opportunity to improve her utility. Let us restate
the definition of arbitrage to make this point explicit.
\begin{defn}
\textsf{(Tradeable arbitrage opportunity)\label{def:Agent-Arb Opportunity}
There is a tradeable arbitrage opportunity in the market with respect
to a probability measure $\mathbf{P}$ if and only if there is a market
participant $i\in I$ who can trade a portfolio $\theta_{i}\in\mathbb{R}^{d}$
with the property
\begin{equation}
\theta_{i}\cdot q\leq0,\,\,\text{but}\,\,\theta_{i}.\tilde{x}\geq0\,\mathbf{P}\text{-a.s.}\,\,\text{and}\,\,\mathbf{P}[\theta_{i}.\tilde{x}>0]>0.\label{eq:Arb defn-2}
\end{equation}
}

\smallskip

We use the term tradeable\emph{ }arbitrage opportunity to distinguish
opportunities that are exploitable. Unless some market participant
can actually trade an arbitrage opportunity that is said to exist,
its presence is contestable at best. The limits to arbitrage literature
(starting with \citealt{key-15})\footnote{\citet{key-26} provide a survey of the limits to arbitrage literature.}
shows that a wedge may exist between Definitions \ref{def:Basic-Arb Opportunity}
and \ref{def:Agent-Arb Opportunity}, and in such cases it is the
latter definition that takes precedence. In the sequel, we shall focus
exclusively on tradeable arbitrage opportunities, and when there is
no chance of confusion we shall drop the prefix ``tradeable'' and
refer to these as just arbitrage opportunities.
\end{defn}
Notice that condition (\ref{eq:Arb defn-2}) is somewhat ambiguous
about the price impact of arbitrage orders. One way to interpret the
definition is that it makes no allowance for any price impact: the
price vector $q$ is fixed no matter what quantities are traded (call
this, the restricted interpretation). In reality, however, trading
almost always entails a price impact.\footnote{An entire sub-area of Finance, the field of Market Microstructure,
is dedicated to understanding the nuances of price impact. See \citet{key-202}
for a survey.} If we would like to take such impact into account, an alternative
way to interpret the arbitrage definition is that the price vector
$q$ in condition (\ref{eq:Arb defn-2}) is portfolio specific. In
other words, if the arbitrage portfolio were $\theta^{'}$ instead
of $\theta$, the ex ante prices would be $q^{'}$ instead of $q$.
In our analysis, we will favor this alternative interpretation of
arbitrage. To highlight the portfolio specific nature of prices, let
us use the notation $q_{\theta}$ to denote the ex-ante prices when
the traded arbitrage portfolio is $\theta$.
\begin{defn}
\textsf{(Tradeable arbitrage opportunity with portfolio specific price)\label{def:With-Price-impact-Arb Opportunity}
There is a tradeable arbitrage opportunity in the market with respect
to a probability measure $\mathbf{P}$ if and only if there is a market
participant $i\in I$ who can trade a portfolio $\theta_{i}\in\mathbb{R}^{d}$
with the property
\begin{equation}
\theta_{i}\cdot q_{\theta_{i}}\leq0,\,\,\text{but}\,\,\theta_{i}.\tilde{x}\geq0\,\mathbf{P}\text{-a.s.}\,\,\text{and}\,\,\mathbf{P}[\theta_{i}.\tilde{x}>0]>0.\label{eq:Arb defn-3}
\end{equation}
}

\smallskip

The only difference between Definitions \ref{def:Agent-Arb Opportunity}
and \ref{def:With-Price-impact-Arb Opportunity} is the use of $q_{\theta_{i}}$
for prices instead of $q$. The ex-ante prices are now transaction
specific: for each portfolio that an agent proposes, the market generates
a customized price. Though the change in notation is minor, Definition
\ref{def:With-Price-impact-Arb Opportunity} provides a much more
general view of arbitrage. Different orders may affect the prices
differently, and market participants have to take the order specificity
of price into account before judging an arbitrage opportunity. In
other words, there is no ``global'' ex-ante price vector for the
assets to begin with, but each arbitrage portfolio $\theta$ generates
its own specific ex-ante price vector $q_{\theta}$. Of course, one
could always recover the restricted interpretation from this general
interpretation by assuming $q_{\theta}=q$ for all $\theta$.

Another related point worth noting is that the original arbitrage
definition purports to look at the arbitrage trade in complete isolation
\textemdash{} the price vector $q$ in (\ref{eq:original price vector})
is simply taken as exogenously given. Nevertheless, an arbitrage opportunity
is in itself a trading opportunity, so any definition of arbitrage
implicitly posits a price formation process. For the original arbitrage
definition, this is a market process where prices stay unchanged no
matter what portfolios are traded, if we go by the restricted interpretation
above. If this were indeed the case, however, prices would never correct
to become arbitrage-free. The source of such conundrums lies in the
lack of an explicit model of a market that is \emph{not} in equilibrium. 

An arbitrage opportunity is not an equilibrium phenomenon, so any
agent exploiting such an opportunity has to build, in her mind, a
model of a market that is independent of the specifics of an equilibrium
notion. Constructing such a model is the focus of the next subsection. 
\end{defn}

\subsection{A Model of the Market, Not Necessarily in Equilibrium, from a Participant's
Perspective \label{subsec:Model of market}}

How should agent $i$ reason about the market and asset prices? To
make this question more concrete, we need a model of the market from
$i$'s perspective. A first version of the model that we employ (this
version does not lay out the belief hierarchies explicitly) is provided
in equation (\ref{eq:strategic form market}) below, and in the intervening
paragraphs we focus on delineating the primitives of this model .

A maintained assumption throughout is that agent $i$ is Bayesian
and she is aware that the market is populated by a finite set $I$
of Bayesian agents. Agent $i$ faces the fundamental uncertainty described
by the measurable space $(S,\mathcal{S})$ and has access to the vector
of $d$ tradeable assets with measurable payoffs $\tilde{x}$ as described
in Section \ref{sec:No-arbitrage}, to navigate the uncertainty. Importantly,
the market that agent $i$ reasons about \emph{need not be in equilibrium}.

Agent $i$ selects a \emph{strategy} from a compact metric space $A_{i}$
in order to express her personal preference for the ex-ante market
prices. Each strategy $a_{i}\in A_{i}$ is a vector representing quantities
of the various assets (i.e., a portfolio) to be bought \textemdash{}
or sold, sell being a negative buy \textemdash{} in a certain sequence.

There is a mapping $\tilde{f}:(a_{i},a_{-i})\mapsto\tilde{m}$, where
$\tilde{m}$ is a bounded random variable on $(S,\mathcal{S})$ representing
the \emph{SDF chosen by the market} after aggregating the individual
market participant choices $(a_{i})_{i\in I}$.\footnote{We are implicitly assuming the ``law of one price'' holds in the
market. This assumption is not strictly required for the results,
but making it simplifies the presentation considerably.}\footnote{We use the standard notation $a_{-i}\in A_{-i}=\prod_{j\neq i}A_{j}$.}
For all the SDFs under consideration, we shall assume that the risk-free
asset has a positive gross return, i.e. $\int_{S}\tilde{m}\mathrm{d}\mathbf{P}>0$.
$\tilde{f}$ is known to all market participants, and the ex-ante
price of the assets prevailing in the market for the transaction can
be backed out from the market SDF $\tilde{m}$ and physical probability
measure $\mathbf{P}$ in the standard manner: $q^{n}=\int_{S}\tilde{m}\tilde{x}^{n}\mathrm{d}\mathbf{P}$,
$n\in\{1,\dots,d\}$. An equilibrium outcome would impose constraints
on the form of $\tilde{f}$, but since arbitrage is not an equilibrium
phenomenon, we shall not concern ourselves here with defining any
special restrictions (like market clearing) for the mapping.

Every strategy that agent $i$ may choose is associated with a \emph{stochastic
net gain}.\footnote{A loss would be a negative gain in our terminology.}
When the agent selects $a_{i}$ as her strategy, her ex-ante cost
is $a_{i}\cdot q_{a_{i}}$. If $\tilde{f}(a_{i},a_{-i})=\tilde{m}_{a_{i}}$,
this cost works out to be $a_{i}\cdot\int_{S}\tilde{m}_{a_{i}}\tilde{x}\mathrm{d}\mathbf{P}$.
To move this cost from the ex-ante to ex-post, the said amount needs
to be invested in a risk free asset. Since we are using portfolio
specific pricing, let us designate the market SDF that results from
investing exclusively in the risk-free asset by $\tilde{m}^{rf}$,
so that an investment of $-a_{i}\cdot q_{a_{i}}$ exclusively into
the risk-free asset invokes the SDF $\tilde{m}_{-a_{i}\cdot q_{a_{i}}}^{rf}$.
Thus, the cost to the agent, ex-post, of choosing strategy $a_{i}$
is $\frac{a_{i}\cdot\int_{S}\tilde{m}_{a_{i}}\tilde{x}\mathrm{d}\mathbf{P}}{\int_{S}\tilde{m}_{-a_{i}\cdot q_{a_{i}}}^{rf}\mathrm{d}\mathbf{P}}$
. The agent's stochastic payout from the strategy is $a_{i}\cdot\tilde{x}$.
Therefore, the stochastic net gain to the agent from selecting strategy
$a_{i}$ is given by the random vector 
\begin{equation}
\tilde{g}(a_{i},a_{-i})=a_{i}\cdot(\tilde{x}-\frac{\int_{S}\tilde{m}_{a_{i}}\tilde{x}\mathrm{d}\mathbf{P}}{\int_{S}\tilde{m}_{-a_{i}\cdot q_{a_{i}}}^{rf}\mathrm{d}\mathbf{P}}).\label{eq:net gain vector}
\end{equation}

An agent's\emph{ utility} depends on the ex-ante prices and ex-post
payoffs from the assets, along with her holding of the assets. Since
the market SDF depends on the strategies of all market participants,
we have $U_{i}:\Pi_{j\in I}A_{j}\times\mathcal{S}\rightarrow\mathbb{R}$
for agent $i$'s utility. Agent $i$ thus reasons that her utility
depends on the choices of the other agents in the market (which collectively
determine the market SDF) and this necessitates strategic thinking
on her part. The strategic thinking entails belief hierarchies which
we take up in subsequent sections. 

The value that agent $i$ derives from a particular contingency $\psi\in\mathcal{S}$
in the probability space $(S,\mathcal{S},\mathbf{P}$) is determined
by her selected strategy's net gain in that contingency. Ceteris paribus,
higher her gain from a strategy in a contingency, greater the utility
she attaches to it, and vice-versa. That is to say, 
\begin{equation}
g(a_{i}',a_{-i},\psi)\,\underset{(=)}{{>}\,}g(a_{i}'',a_{-i},\psi)\Longleftrightarrow U_{i}(a_{i}',a_{-i},\psi)\underset{(=)}{\,{>}\,}U_{i}(a_{i}'',a_{-i},\psi).\label{eq:util and action relation}
\end{equation}

Using the primitives described above, from agent $i$'s perspective
the market may be represented by the tuple 
\begin{equation}
\mathcal{M}_{i}=\left((S,\mathcal{S}),\tilde{x},I,(A_{j})_{j\in I},\tilde{f},U_{i}\right).\label{eq:strategic form market}
\end{equation}
Given $\mathcal{M}_{i}$, the agent makes her choice $a_{i}\in A_{i}$
and is told the aggregate market selection in the form of an SDF $\tilde{m}$.
Notice that beyond the fact that there are $I$ market participants
who are making selections from their respective sets $A_{j}$, agent
$i$ is not presumed to know anything about other market agents. 

One can recover the regular ``price-quantity'' version of the market
by defining $\mathcal{M}=(\mathcal{M}_{i})_{i\in I}$, and using the
market SDF with $\tilde{f}$ to obtain prices, and the strategies
$a_{j}$, $j\in I$, to obtain the quantities traded. An equilibrium
notion imposes its own restrictions on the elements of $\mathcal{M}$
(for example, market clearing and utility maximization in the Walrasian
equilibrium), but it is not necessary to impose these restrictions
for an analysis of non-equilibrium phenomena like arbitrage. $\mathcal{M}_{i}$
in (\ref{eq:strategic form market}) is not linked to any specific
equilibrium notion or solution concept: it is just a description of
the market from agent $i$'s perspective.

\subsection{Arbitrage in $\mathcal{M}_{i}$ \label{subsec:Arbitrage-in-M}}

We would like to extend the definition of arbitrage to the setup of
$\mathcal{M}_{i}$, and this is accomplished in Proposition \ref{prop:Amenable-Arb-def}
below. Let $a_{i}$ denote agent $i$'s initial choice of strategy.
Recall that the mapping $\tilde{f}$ aggregates all market participant
strategies into the SDF $\tilde{m}$. Since $\tilde{f}$ could be
many-one, the inverse mapping $\tilde{f}^{-1}$ is a correspondence.
So, if agent $i$ sees that the market SDF is $\tilde{m}$ when she
has chosen $a_{i}$, she can back out a set of strategies for other
market participants $A_{-i}(\tilde{m},a_{i})\subseteq\prod_{j\neq i}A_{j}$
which satisfies the condition 
\begin{equation}
a_{i}\times A_{-i}(\tilde{m},a_{i})=\tilde{f}^{-1}(\tilde{m}).\label{eq:A_i(m,a_i)}
\end{equation}
That is to say, $A_{-i}(\tilde{m},a_{i})$ denotes the set of strategies
that might have been used by other market participants according to
agent $i$, given she chose $a_{i}$ and observed the market SDF $\tilde{m}$. 

Agent $i$ has a tradeable arbitrage opportunity if she can find an
alternative strategy that increases her utility weakly in every plausible
state of nature \textemdash{} and strictly in at least some states
\textemdash{} given the strategies of other market participants she
has backed out from asset prices. The utility based characterization
of arbitrage in Proposition \ref{prop:Amenable-Arb-def} serves as
a useful stepping stone on the road towards more general descriptions
of arbitrage based on belief hierarchies.
\begin{prop}
\label{prop:Amenable-Arb-def} There is a tradeable arbitrage opportunity
in the market with respect to a probability measure $\mathbf{P}$
if and only if there is a market participant $i\in I$ who may change
her strategy from $a_{i}$ to $a_{i}^{*}$, such that\emph{
\begin{equation}
\tilde{U}_{i}(a_{i}^{*},a_{-i})\geq\tilde{U}_{i}(a_{i},a_{-i})\,\mathbf{P}\text{-a.s.}\,\,\text{and}\,\,\mathbf{P}[\tilde{U}_{i}(a_{i}^{*},a_{-i})>\tilde{U}_{i}(a_{i},a_{-i})]>0,\label{eq:arbitrage utilities}
\end{equation}
}for all $a_{-i}\in A_{-i}(\tilde{m},a_{i})$.
\end{prop}
\begin{proof}
In the Appendix.
\end{proof}
\smallskip

The definition of arbitrage takes on a particularly simple form in
(\ref{eq:arbitrage utilities}). There is an arbitrage opportunity
if and only if a market participant can revise her strategy to generate
weakly higher utility in every state of nature and strictly higher
utility in some states. In a concrete sense, therefore, an arbitrage
opportunity implies that there is a market participant $i$ who is
responding to the aggregate strategy of other participants with a
strategy that can be improved. The aggregate strategy of all participants
in the market is embodied in the ex-ante prices, and this gives the
equivalence between (\ref{eq:arbitrage utilities}) and earlier definitions
of arbitrage. The formulation of arbitrage in Proposition \ref{prop:Amenable-Arb-def}
is handy for us because it rests on strategic foundations (the market
SDF depends on the choice of all market participants) despite bypassing
the intricacies of equilibrium formation.

Having characterized an arbitrage opportunity in terms of agent $i$'s
strategies and utility, we can now approach the question in terms
of ``dominated'' and ``undominated'' responses. Agent $i$'s strategy,
$a_{i}$, is said to be a dominated response to $A_{-i}(\tilde{m},a_{i})$
if there is an $a_{i}^{*}\in A_{i}$ that satisfies condition (\ref{eq:arbitrage utilities})
in the Proposition above. We shall label such strategies \emph{dominated-wrtp
responses}, where the acronym w.r.t.p. stands for ``with respect
to price''. The label comes from the fact that when agent $i$ employs
such a strategy, she can improve her response given the resulting
prices. 
\begin{defn}
\textsf{(Dominated/Undominated-wrtp response) \label{def:Dominated/Undominated-wrtp-price}Given
her current strategy $a_{i}\in A_{i}$ and other market participant
strategies $a_{-i}\in A_{-i}(\tilde{m},a_{i})$, agent $i$ is said
to be using a }\textsf{\emph{dominated-wrtp reponse}}\textsf{ if she
can use another strategy $a_{i}^{*}\neq a_{i}$ such that $a_{i}^{*}$
and $a_{i}$ satisfy condition (\ref{eq:arbitrage utilities}) in
Proposition \ref{prop:Amenable-Arb-def}. It given by the set 
\begin{flalign}
D_{i}^{wrtp}= & \{a_{i}\in A_{i}:\text{Given \ensuremath{a_{-i}\in A_{-i}(\tilde{m},a_{i})}, there exists }a_{i}^{*}\in A_{i}\text{ s.t. }\nonumber \\
 & \text{ }\tilde{U}_{i}(a_{i}^{*},a_{-i})\geq\tilde{U}_{i}(a_{i},a_{-i})\,\mathbf{P}\text{-a.s.}\,\,\text{and}\,\,\mathbf{P}[\tilde{U}_{i}(a_{i}^{*},a_{-i})>\tilde{U}_{i}(a_{i},a_{-i})]>0\}.\label{eq:dominated wrtp response}
\end{flalign}
}

\textsf{A strategy that is not dominated-wrtp is termed an }\textsf{\emph{undominated-wrtp
response. }}\textsf{It is given by the set
\[
UD_{i}^{wrtp}=\{a_{i}\in A_{i}:a_{i}\notin D_{i}^{wrtp}\}.
\]
}
\end{defn}
\smallskip

The set $A_{-i}(\tilde{m},a_{i})\subseteq\prod_{j\neq i}A_{j}$ plays
an important role in our analysis. Though we've labeled the responses
in Definition \ref{def:Dominated/Undominated-wrtp-price} dominated
with respect to price, in reality these are dominated with respect
to the market participant strategies that support the price; i.e.,
the set $A_{-i}(\tilde{m},a_{i})$. In these terms, Proposition \ref{prop:Amenable-Arb-def}
says that there is a tradeable arbitrage opportunity in the market
only when a market participant $i\in I$ is using a dominated-wrtp
response.
\begin{cor}
\emph{(Proposition \ref{prop:Amenable-Arb-def})}.\label{cor:Corollary Amenable Arb}
There is a tradeable arbitrage opportunity in the market with respect
to a probability measure $\mathbf{P}$ if and only if there is a market
participant $i\in I$ who is using a dominated-wrtp response.
\end{cor}
We could, in fact, define subsets of the strategy space $\prod_{j\neq i}A_{j}$
using many different criteria. As we will see when we define agent
$i$'s hierarchy of beliefs in the next section, at each level of
such a hierarchy there is a different subset of $\prod_{j\neq i}A_{j}$
that agent $i$ could deem plausible for other market participants,
and each case generates a different set of dominated responses. To
decide whether or not the new dominated response sets coincide with
the agent's dominated-wrtp response, we will have to compare $A_{-i}(\tilde{m},a_{i})$
with the subsets of $\prod_{j\neq i}A_{j}$ generated by the belief
hierarchy. A tradeable arbitrage opportunity is equivalent to a dominated-wrtp
response, so this creates a way to bridge the classical approach to
the study of arbitrage with a belief hierarchy based approach.

\section{A Bayesian's Belief Hierarchy \label{sec:Hierarchies}}

In this section, we continue to reason from the standpoint of a particular
market participant, but instead of SDF as primitive, we introduce
the belief hierarchies supporting the SDF as the basis of our model.
A takeaway is that the space of uncertainty that agent $i$ actually
faces in the market is much larger than $S$ \textemdash{} the domain
of fundamental uncertainty \textemdash{} that has been been the traditional
focus in asset pricing. This is because agent $i$ is uncertain not
just about fundamentals, but also about the choices made by other
market participants, about what other market participants believe
her choice to be, about what other market participants believe about
what she believes their choices to be, and so on ad infinitum. Despite
the seemingly limitless size, such uncertainty spaces are well-studied
mathematical objects, and an agent navigates them by assigning plausibility
to certain subsets of belief hierarchies over others. In this section,
we define such a plausible hierarchy of beliefs, $W_{i}^{k}$, that
plays a special role in the analysis of arbitrage opportunities. We
characterize some of the salient properties of $W_{i}^{k}$ and define
the notion of a dominated $k^{th}$ order response for this hierarchy.
In Section \ref{sec:No-arbitrage-and-Higher-Beliefs}, we shall use
the hierarchy $W_{i}^{k}$ to provide a belief based foundation for
an arbitrage opportunity.

We continue to use the market model described by condition (\ref{eq:strategic form market})
for agent $i$. However, the focus in this Section will be primarily
on how agent $i$ reasons about the choice of strategy of other market
participants, $j\neq i$, from $(A_{j})_{j\neq I}$ \emph{before}
she has herself chosen a strategy $a_{i}$ and received the aggregate
market choice $\tilde{m}$. Thus, we shall be discussing the reasoning
process employed by the agent \emph{in the} \emph{absence of any signal
from the market}. Recall that agent $i$'s utility depends on $\tilde{x}$
and $\tilde{f}$ through the relation in (\ref{eq:util and action relation}).
For this Section, we shall abstract away from this dependence and
assume directly that agent $i$ knows the utility mapping $U_{i}:\Pi_{j\in I}A_{j}\times\mathcal{S}\rightarrow\mathbb{R}$.
Effectively, therefore, we shall be using only the reduced form 

\begin{equation}
\mathbb{M}_{i}=\left((S,\mathcal{S}),I,(A_{j})_{j\in I},U_{i}\right),\label{eq:reduced form market}
\end{equation}
of the original market model $\mathcal{M}_{i}$ for agent $i$. 

\subsection{Belief Hierarchies and Canonical Homeomorphism\label{subsec:Canonical}}

This subsection describes the construction of belief hierarchies for
agent $i$. While not common in finance, the construction below is
standard in the epistemic game theoretic literature (see \citealt{key-11}).
A reader who is familiar with that literature may skip ahead to the
next subsection after browsing through our notation.

How should a Bayesian market participant reason about her situation?
To start with, agent $i$ is uncertain about the state of nature $S$.
However, all our definitions of arbitrage presume the physical probability
measure $\mathbf{P}$ is known to market participants, so agent $i$'s
prior belief over $(S,\mathcal{S})$ is predetermined. We label $b_{i}^{0}=\mathbf{P}$
agent $i$'s \emph{zeroth order belief. $b_{i}^{0}$} is a member
of the singleton set $B_{i}^{0}$, the set of permitted zeroth order
beliefs for agent $i$.

Agent $i$ is also uncertain about the strategies that are chosen
by other market participants. $S$ and $(A_{j})_{j\neq i}$ together
determine the state space for agent $i$'s layer-0 uncertainty 

\begin{equation}
Y_{i}^{0}=S\times\prod_{j\neq i}A_{j}
\end{equation}
and, as a Bayesian, she must have a prior belief on this space. $Y_{i}^{0}$
is a compact metric space since $S$ and $A_{j}$ are compact metric
spaces, and we let $\Delta(Y_{i}^{0})$ denote the set of probability
measures on the Borel $\sigma$-field of $Y_{i}^{0}$ endowed with
the topology of weak convergence. Then $B_{i}^{1}=\Delta(Y_{i}^{0})$
is again a compact metric space. Agent $i$'s \emph{first order belief}
$b_{i}^{1}$ is a member of the set $B_{i}^{1}$. 

Agent $i$ realizes that as Bayesians, all agents in the market carry
a first order belief in their head, but she is uncertain about the
rest of the market's first order beliefs. Her second order belief
is a prior over the first order beliefs of other agents in the market
and her own layer-0 uncertainty. Iterating such arguments, we get
the state space for agent $i$'s layer-$k$ uncertainty 
\begin{equation}
Y_{i}^{k}=Y_{i}^{k-1}\times\prod_{j\neq i}B_{j}^{k}.
\end{equation}
The agent's $(k+1)^{th}$ \emph{order belief} $b_{i}^{k+1}$ is a
member of the set $B_{i}^{k+1}=\Delta(Y_{i}^{k})$. 

Notice that given such a \emph{hierarchy of beliefs}, an agent may
compute the probability of an event in multiple ways. For instance,
both $b_{i}^{1}$ and $\marg{}_{Y_{i}^{0}}b_{i}^{2}$ give agent $i$'s
beliefs on $S\times\prod_{j\neq i}A_{j}$. Belief hierarchies are
termed \emph{coherent} when they lead to the same probability for
events, no matter how the probability is calculated; i.e., when for
all $k>1$,
\begin{equation}
\marg{}_{Y_{i}^{k-1}}b_{i}^{k+1}=b_{i}^{k}.\label{eq:Coherent beliefs}
\end{equation}
As is standard, we shall assume that not only does agent $i$ have
coherent beliefs, but also every other agent in the set $I$ has coherent
beliefs, and this fact is common knowledge in the market. We shall
use the label \emph{consistent} to denote beliefs that reflect coherence
and common knowledge of coherence of beliefs in the market.\footnote{\postdisplaypenalty=10000 To define consistent hierarchies, let $H_{i}$
denote the set of coherent belief hierarchies for agent $i$. A consistent
belief hierarchy for agent $i$ is the set 
\[
B_{i}=\{(b_{i}^{0},b_{i}^{1},b_{i}^{2},\dots)\in H_{i}:\marg_{B_{j}^{k-1}}b_{i}^{k}(b_{j}^{k-1})=1\,\,\forall k\geq1,\forall j\neq i\,\,\text{and\,\,(\ensuremath{b_{j}^{0}},\ensuremath{b_{j}^{1}},\ensuremath{b_{j}^{2}},\ensuremath{\dots})\ensuremath{\in H_{j}}}\}.
\]
For more detailed definitions of consistency, refer \citet{key-20}
or \citet{key-9}.} Coherence and consistency of beliefs are important properties of
a belief hierarchy, and a maintained assumption throughout the sequel
shall be that the beliefs under consideration satisfy these two conditions. 

Agent $i$'s belief hierarchy $b_{i}$ is, therefore, a point in the
space
\begin{equation}
B_{i}=\{(b_{i}^{0},b_{i}^{1},b_{i}^{2},\dots)\in\prod_{k\geq0}B_{i}^{k}:(b_{i}^{0},b_{i}^{1},b_{i}^{2},\dots)\text{ is consistent}\}.\label{eq:Consistent beliefs}
\end{equation}
The set $B_{i}$ in (\ref{eq:Consistent beliefs}) is a compact and
metric subset of $\prod_{k\geq0}B_{i}^{k}$ under the product topology,
and foundational work in epistemic game theory (\citealt{key-24},
\citealt{key-19}, \citealt{key-10}, \citealt{key-9}, \citealt{key-27},
among others) has shown that the sets $B_{i}$ and $\Delta(Y_{i}^{0}\times\prod_{j\neq i}B_{j})$
are homeomorphic. Hence, these two sets are of the ``same size''
and agent $i$ does not need to consider any further priors. Further,
the Daniell-Kolmogorov existence theorem and its extensions (see \citealt{key-8},
Chapter 3) guarantee that the homeomorphism

\begin{equation}
\phi_{i}:B_{i}\rightarrow\Delta(Y_{i}^{0}\times\prod_{j\neq i}B_{j})
\end{equation}
is \emph{canonical, }with the property that for $b_{i}\in B_{i}$,
$\marg{}_{Y_{i}^{k-1}}[\phi_{i}(b_{i})]=b_{i}^{k}$. 

In sum, the state space for the \emph{total domain of uncertainty}
faced by agent $i$ in the market is $\Omega=S\times\prod_{j\neq i}A_{j}\times\prod_{j\neq i}B_{j}$.
This is a much larger space than the domain of fundamental uncertainty,
$S$, that is normally used in standard asset pricing theories. Yet,
only by accounting for the space of belief hierarchies, $B_{i}$ in
equation (\ref{eq:Consistent beliefs}), does one exhaust all the
uncertainty faced by agent $i$. Since finite ordinals are enough
to characterize the hierarchies, standard mathematical induction suffices
to characterize their attributes. In the next subsection we define
a subset of the consistent belief hierarchies that is pertinent to
our characterization of arbitrage, and derive some of its properties
via induction. 

\subsection{Dominated and Undominated Responses and Belief Hierarchy Set $W_{i}^{k}$
\label{subsec:Dominated-Responses}}

Consistent belief hierarchies, in themselves, are too broad to be
useful. Consequently, one imposes further natural restrictions on
such hierarchies depending on the application at hand. This is the
motivation behind epistemic solution techniques like rationalizability
(\citealt{key-5}, \citealt{key-3}) or rationality and common knowledge
of rationality (\citealt{key-21}, \citealt{key-20}). For our application
\textemdash{} investigating the link with arbitrage \textemdash{}
it is simpler to use a notion of optimization that is a tad bit different
from conventional rationality. In conventional rationality under uncertainty,
rational agents are deemed to choose strategies that maximize their
expected subjective utility. In our case, optimizing agents shall
be deemed to choose strategies that give them a weakly higher utility
in every state of nature and strictly higher utility in at least some
states. This leads to the belief hierarchy set $W_{i}^{k}$.

\begin{table}[t]
\caption{\textsf{\small{}\label{tab:Table Undominated Dominated Belief hierarchy}Dominated
and Undominated responses, and Belief hierarchy set $W_{i}^{k}$}}

\smallskip{}

\centering{}%
\begin{tabular}{>{\centering}p{0.09\paperwidth}>{\raggedright}p{0.65\paperwidth}}
\toprule 
\textsf{\small{}Condition} & \textsf{\small{}Short explanation of the condition}\tabularnewline
\midrule
\midrule 
\textsf{\small{}$a_{i}\in\mathcal{D}_{i}^{1}$ $a_{i}\in\mathcal{UD}_{i}^{1}$} & \textsf{\small{}Agent $i$'s strategy $a_{i}$ is in the }\textsf{\emph{\small{}dominated
first order response set}}\textsf{\small{} $\mathcal{D}_{i}^{1}$
if she can unequivocally improve her utility by choosing an alternative
response, no matter what strategies other market participants use.
Her strategy is in the }\textsf{\emph{\small{}undominated first order
response set}}\textsf{\small{} $\mathcal{UD}_{i}^{1}$ if it is not
a member of the set $\mathcal{D}_{i}^{1}$.}\tabularnewline
\midrule 
\textsf{\small{}$a_{i}\in\mathcal{D}_{i}^{k}$, $k>1$} & \textsf{\small{}Agent $i$'s strategy $a_{i}$ is in the }\textsf{\emph{\small{}dominated
$k^{th}$ order response set}}\textsf{\small{} $\mathcal{D}_{i}^{k}$
if it is a member of her undominated $(k-1)^{th}$ order response
set, and she can unequivocally improve her utility by choosing an
alternative response when her beliefs about other market participants'
strategies come from a belief hierarchy $b_{i}\in W_{i}^{k}$.}\tabularnewline
\midrule 
\textsf{\small{}$a_{i}\in\mathcal{UD}_{i}^{k}$, $k>1$} & \textsf{\small{}Agent $i$'s strategy $a_{i}$ is in the }\textsf{\emph{\small{}undominated
$k^{th}$ order response set}}\textsf{\small{} $\mathcal{UD}_{i}^{k}$
if it is a member of her undominated $(k-1)^{th}$ order response
set, and she cannot unequivocally improve her utility by choosing
an alternative response when her beliefs about other market participants'
strategies come from a belief hierarchy $b_{i}\in W_{i}^{k}$.}\tabularnewline
\midrule 
\textsf{\small{}$b_{i}\in W_{i}^{k}$, $k>1$} & \textsf{\small{}Agent $i$ is using a }\textsf{\emph{\small{}belief
hierarchy $b_{i}$ in the set $W_{i}^{k}$}}\textsf{\small{} means:}{\small\par}
\begin{itemize}[nolistsep]
\item \textsf{\small{}She believes that other agents in the market are
using strategies in their undominated $(k-1)^{th}$ order response
sets}{\small\par}
\item \textsf{\small{}She believes that other agents in the market believe
that she is using a strategy in her undominated $(k-1)^{th}$ order
response set}{\small\par}
\end{itemize}
\textsf{\small{}Agent $i$ uses a belief hierarchy $b_{i}\in W_{i}^{k}$
only if her own strategy is a member of her undominated $(k-1)^{th}$
order response set.}\tabularnewline
\bottomrule
\end{tabular}
\end{table}

In this section, we provide a microfoundation for the belief hierarchy
set $W_{i}^{k}$. We describe this hierarchy as the subset of consistent
belief hierarchies for agent $i$ that imposes two additional restrictions:
(i) the agent believes that other participants do not use dominated
responses, (ii) the agent believes that other market participants
believe that she does not use dominated responses. In order that her
belief about the belief of other market participants be valid, the
agent uses a belief hierarchy in $W_{i}^{k}$ only when she does not
actually use a dominated response. The definition of dominated and
undominated response sets use our notion of optimizing agents instead
of conventional rationality. Finally, we derive a few key properties
of this hierarchy that are useful for the subsequent analysis. Table
\vref{tab:Table Undominated Dominated Belief hierarchy} gives an
intuitive summary of the main characteristics of the belief hierarchy
set $W_{i}^{k}$, and dominated and undominated responses on this
hierarchy. We use the term ``unequivocally improve utility'' in
the table as a shorthand for utility that weakly increases in every
state of nature in the support of $\mathbf{P}$, and strictly increases
in at least some states. Readers primarily interested in the arbitrage
side of the story may skip ahead to Section \ref{sec:No-arbitrage-and-Higher-Beliefs}
after perusing the table if they so prefer, and come back to this
section for details as needed.

Recall that each belief $b_{i}^{k}$ in agent $i's$ hierarchy is
a probability measure that has a support \textemdash{} the states
to which the measure assigns non-zero probability.\footnote{\label{fn:Support}The support of a probability measure $\mathbf{Q}$
on a measurable space $(B,\mathcal{B})$, denoted by supp($\mathbf{Q}$),
is the smallest closed subset $\bar{B}$ of $B$ such that $\mathbf{Q}(\bar{B})=1$.} The probability measure induced by the entire hierarchy of beliefs
is given by $\phi_{i}(b_{i})$, and since the canonical homeomorphism
$\phi_{i}$ ``preserves beliefs,'' agent $i$ can recover her $k^{th}$
order belief from $\phi_{i}$ by taking the appropriate marginal;
i.e., ${\marg_{Y_{i}^{k-1}}{[\phi_{i}(b_{i})]}}=b_{i}^{k}.$\footnote{$\marg_{X}{[\mathbf{Q}]}$ is the marginal of probability measure
$\mathbf{Q}$ on set $X$. $\supp{\marg_{X}{[\mathbf{Q}]}}$ is the
support of the marginal of $\mathbf{Q}$ on $X$.} 

By a\emph{ }dominated response in the context of a belief hierarchy,
we mean, roughly, a strategy of agent $i$ that can be improved no
matter what strategies other market participants follow and which
state of nature realizes ex-post, given they are in the support of
agent $i$'s beliefs. Dominated responses with respect to belief hierarchies
are closely related to dominated-wrtp responses and arbitrage opportunities,
and we shall establish the connection rigorously in subsequent sections.
For now, to give a precise meaning to a dominated response in the
context of belief hierarchies, we need to provide an inductive definition
for the notion. 

Agent $i$'s strategy is dominated with respect to her zeroth order
belief if she has an alternative strategy that can increase her utility
weakly in every state of nature in the support of the probability
measure $b_{i}^{0}=\mathbf{P}$, no matter what strategies the other
market participants select, and further, in at least some states of
nature the alternative increases her utility strictly. Such strategies
are labeled dominated first order responses. An undominated response
is a strategy that is not dominated. Thus, an undominated first order
response by agent $i$ to her zeroth order belief cannot be weakly
improved upon in every state of nature that agent $i$ presumes possible
as well as strictly improved upon in some states, for every strategy
that other market participants can use. 
\begin{defn}
\textsf{(Dominated/Undominated first order response)\label{def:SD level-1}
Given a probability measure $\mathbf{P}$ over states of nature, agent
$i\in I$ is said to be using a}\textsf{\emph{ dominated first order
response}}\textsf{ if her strategy lies in the set 
\begin{flalign}
\mathcal{D}_{i}^{1}= & \{a_{i}\in A_{i}:\text{For all \ensuremath{a_{-i}\in\prod_{j\neq i}A_{j},} there exists }a_{i}^{*}\in A_{i}\text{ s.t. }\nonumber \\
 & \tilde{U}_{i}(a_{i}^{*},a_{-i})\geq\tilde{U}_{i}(a_{i},a_{-i})\,\mathbf{P}\text{-a.s.}\text{ and}\,\,\mathbf{P}[\tilde{U}_{i}(a_{i}^{*},a_{-i})>\tilde{U}_{i}(a_{i},a_{-i})]>0\}.
\end{flalign}
 A strategy that is not a dominated response to agent $i$'s zeroth
order belief is termed an }\textsf{\emph{undominated first order response}}\textsf{.}\textsf{\emph{
}}\textsf{It is given by the set
\begin{equation}
\mathcal{UD}_{i}^{1}=\{a_{i}\in A_{i}:a_{i}\notin\mathcal{D}_{i}^{1}\}.\label{eq:undominated W_i^1}
\end{equation}
}
\end{defn}
\smallskip

Since dominated first order responses can be unequivocally improved,
it would seem reasonable to suppose that agent $i$ should not believe
that other market participants shall employ such strategies. In other
words, agent $i$ should employ a belief hierarchy in 
\begin{equation}
V_{i}^{2}=\{b_{i}\in B_{i}:\marg_{A_{j}}{[\phi_{i}(b_{i})]}[a_{j}\in\mathcal{D}_{j}^{1}]=0\text{ for all }j\neq i\},\label{eq:set V_i^2}
\end{equation}
since this set excludes beliefs in $B_{i}$ that give a non-zero probability
weight to the use of dominated strategies in $\mathcal{D}_{j\neq i}^{1}$
by other participants. By the same token, agent $i$ must also anticipate
that other agents would not believe that she has employed a strategy
in $\mathcal{D}_{i}^{1}$. In this case, the agent should employ a
belief hierarchy in 
\begin{equation}
W_{i}^{2}=\{b_{i}\in V_{i}^{2}:b_{j}^{1}\in\supp{\marg_{B_{i}^{2}}{[\phi_{i}(b_{i})]}}\Longrightarrow\marg_{A_{i}}{[b_{j}^{1}]}[a_{i}\in D_{i}^{1}]=0\text{ for all }j\neq i\},\label{eq:set W_i^2}
\end{equation}
since this set excludes the beliefs in $V_{i}^{2}$ that permit agent
$i$ to believe that other market participants believe that $i$ has
used a dominated strategy in $\mathcal{D}_{i}^{1}$. It also seems
reasonable to suppose that if agent $i$ is indeed using a belief
hierarchy in $W_{i}^{2}$, then she should have selected her strategy
from $\mathcal{UD}_{i}^{1}$. That is,
\begin{equation}
b_{i}\in W_{i}^{2}\implies a_{i}\in\mathcal{UD}_{i}^{1}.\label{eq:restriction from hierarchy to action}
\end{equation}
Condition (\ref{eq:restriction from hierarchy to action}) ensures
that agent $i$'s belief about other market participants' beliefs
(about $i$'s strategy) are valid. One would like to rule out invalid
presumptions from $i$'s belief hierarchy because utility gains supported
by invalid beliefs would not actually fructify for the agent.

The procedure can be extended iteratively to provide a general inductive
definition for $k^{th}$ order belief hierarchy sets. Let $\mathcal{UD}_{i}^{k}\subseteq\mathcal{UD}_{i}^{k-1}\subseteq\dots\subseteq\mathcal{UD}_{i}^{2}\subseteq\mathcal{UD}_{i}^{1}$
be a series of subsets of $\mathcal{UD}_{i}^{1}$, $i\in I$. We shall
microfound these subsets as higher order undominated responses shortly,
but for now let us just take these to be nested subsets of $\mathcal{UD}_{i}^{1}$.
Let $\mathcal{D}_{i}^{2}=\text{\ensuremath{\mathcal{UD}}}_{i}^{1}\setminus\mathcal{UD}_{i}^{2}$,
and more generally $\mathcal{D}_{i}^{k}=\text{\ensuremath{\mathcal{UD}}}_{i}^{k-1}\setminus\mathcal{UD}_{i}^{k}$.
We can define the belief hierarchy set $W_{i}^{k}$ almost exactly
as we defined $W_{i}^{2}$ in the foregoing paragraph. In the first
stage define the set
\begin{equation}
V_{i}^{k}=\{b_{i}\in W_{i}^{k-1}:\marg_{A_{j}}{[\phi_{i}(b_{i})]}[a_{j}\in\mathcal{D}_{j}^{k-1}]=0\text{ for all }j\neq i\},\label{eq:Belief hierarchy set V_i^k}
\end{equation}
which excludes the beliefs in $W_{i}^{k-1}$ that give a non-zero
probability weight to the use of dominated strategies in $\mathcal{D}_{j\neq i}^{k-1}$
by other participants. Next, define $W_{i}^{k}$ as the subset of
$V_{i}^{k}$ with the following property
\begin{equation}
W_{i}^{k}=\{b_{i}\in V_{i}^{k}:b_{j}^{k-1}\in\supp{\marg_{B_{i}^{k}}{[\phi_{i}(b_{i})]}}\Longrightarrow\marg_{A_{i}}{[b_{j}^{k-1}]}[a_{i}\in D_{i}^{k-1}]=0\text{ for all }j\neq i\}.\label{eq:Belief hierarchy set W_i^k}
\end{equation}
As with $W_{i}^{1}$, the set $W_{i}^{k}$ excludes the beliefs in
$V_{i}^{k}$ that permit agent $i$ to believe that other market participants
believe that $i$ has used a dominated strategy in $\mathcal{D}_{i}^{k-1}$.
Finally, if agent $i$ is indeed using a belief hierarchy in $W_{i}^{k}$,
then she should have actually selected her strategy from $\mathcal{UD}_{i}^{k-1}$,
to ensure agent $i$'s belief about other market participants' beliefs
(about $i$'s strategy) are valid. That is,
\begin{equation}
b_{i}\in W_{i}^{k}\implies a_{i}\in\mathcal{UD}_{i}^{k-1}.\label{eq:restriction_general hierarchy to action}
\end{equation}

Conditions (\ref{eq:Belief hierarchy set V_i^k}) \textendash{} (\ref{eq:restriction_general hierarchy to action})
characterize the belief hierarchy set $W_{i}^{k}$. Notice that this
is an inductive definition \textemdash{} the definition of $W_{i}^{k}$
depends on the definition of $W_{i}^{k-1}$$\dots$ depends on the
definition of $W_{i}^{2}$. The set $W_{i}^{2}$ is a subset of $B_{i}$,
the set of consistent belief hierarchies, and we adopt the convention
$W_{i}^{1}=B_{i}$

The definition of $W_{i}^{k}$ also rests on the definitions of dominated
and undominated sets. We have already defined $\mathcal{D}_{i}^{1}$
and $\mathcal{UD}_{i}^{1}$ in Definition \ref{def:SD level-1}. We
now proceed to define the higher order dominated and undominated sets.
Intuitively, a dominated $k^{th}$ order response for agent $i$ is
dominated with respect to her $(k-1)^{th}$ order belief about other
market participant strategies and states of nature. That is to say,
she has an alternative response that can increase her utility weakly
in every state of nature in the support of the probability measure
$\mathbf{P}$ \textemdash{} given the strategies of the other market
participants are chosen from the support of the marginal of $b_{i}^{k-1}$
on $\prod_{j\neq i}A_{j}$ \textemdash{} and further, in at least
some states of nature the alternative increases her utility strictly.
The belief hierarchies used for dominated and undominated $k^{th}$
order responses come from the set $W_{i}^{k}$.

\begin{defn}
\label{def:SD level k}\textsf{(Dominated/Undominated $k^{th}$ order
response on $W_{i}^{k}$)} \textsf{Given belief hierarchy $b_{i}\in W_{i}^{k}$
for agent $i$, she is said to be using a}\textsf{\emph{ dominated
$k^{th}$ order response }}\textsf{if her strategy is a member of
the set 
\begin{flalign}
\mathcal{D}_{i}^{k}= & \{a_{i}\in\mathcal{UD}{}_{i}^{k-1}:\text{For all \ensuremath{a_{-i}\in\supp{\marg_{\prod_{j\neq i}A_{j}}{[\phi_{i}(b_{i})]}},} there exists }a_{i}^{*}\in A_{i}\text{ s.t.}\nonumber \\
 & \tilde{U}_{i}(a_{i}^{*},a_{-i})\geq\tilde{U}_{i}(a_{i},a_{-i})\,\mathbf{P}\text{-a.s.}\text{ and}\,\,\mathbf{P}[\tilde{U}_{i}(a_{i}^{*},a_{-i})>\tilde{U}_{i}(a_{i},a_{-i})]>0\},\label{eq:dominated W^k}
\end{flalign}
A strategy that is not a dominated response to agent $i$'s $k^{th}$
order belief is termed an }\textsf{\emph{undominated $k^{th}$ order
response. }}\textsf{It is given by the set
\begin{flalign}
\mathcal{UD}{}_{i}^{k}= & \{a_{i}\in\mathcal{UD}{}_{i}^{k-1}:a_{i}\notin\mathcal{D}_{i}^{k}\}.\label{eq:undominated W^k}
\end{flalign}
}
\end{defn}
\smallskip

\begin{figure}[t]
\begin{tikzpicture}[domain=0:4][font=\small]  
\draw[step=0.1,very thin,color=gray!20] (-4,-0.5) grid (12.5,5); 
\draw[very thick][-] (-1.5,0) -- (10.5,0);
\node[font=\normalsize] [above] at (4.5,0){$A_i=\mathcal{UD}_i^{0}$};
\draw[thick][color=black][-] (-1.5,-0.25) -- (-1.5,0.25);
\draw[thick][color=black][-] (10.5,-0.25) -- (10.5,0.25);
\draw[very thick][-] (-1.5,1) -- (10.5,1);
\node[font=\normalsize] [above] at (5.5,1){$\mathcal{UD}_i^1$};
\node[font=\normalsize] [above] at (-0.5,1){$\mathcal{D}_i^1$};
\draw[thick][color=black][-] (-1.5,0.75) -- (-1.5,1.25);
\draw[thick][color=black][-] (0.5,0.75) -- (0.5,1.25);
\draw[thick][color=black][-] (10.5,0.75) -- (10.5,1.25);
\draw[very thick][-] (0.5,2) -- (10.5,2);
\node[font=\normalsize] [above] at (6.5,2){$\mathcal{UD}_i^2$};
\node[font=\normalsize] [above] at (1.5,2){$\mathcal{D}_i^2$};
\draw[thick][color=black][-] (0.5,1.75) -- (0.5,2.25);
\draw[thick][color=black][-] (2.5,1.75) -- (2.5,2.25);
\draw[thick][color=black][-] (10.5,1.75) -- (10.5,2.25);
\draw[very thick][-] (2.5,3) -- (10.5,3);
\node[font=\normalsize] [above] at (7.5,3){$\mathcal{UD}_i^3$};
\node[font=\normalsize] [above] at (3.5,3){$\mathcal{D}_i^3$};
\draw[thick][color=black][-] (2.5,2.75) -- (2.5,3.25);
\draw[thick][color=black][-] (4.5,2.75) -- (4.5,3.25);
\draw[thick][color=black][-] (10.5,2.75) -- (10.5,3.25);
\node[font=\large] [above] at (3.5,3.75){\vdots};
\node[font=\large] [above] at (7.5,3.75){\vdots};
\end{tikzpicture}
\caption{ \textsf{\bf Pictorial representation of dominated and undominated responses.} \textsf{ The dominated response sets in the hierarchy don't overlap while the undominated response sets form a nested hierarchy.}}\label{fig:Domundom}
\end{figure}

Recall that the set $\mathcal{UD}_{i}^{1}$ is a subset of $A_{i}$,
and we adopt the convention $\mathcal{UD}_{i}^{0}=A_{i}$. Figure
\vref{fig:Domundom} provides a pictorial representation of dominated
and undominated responses. Notice that the undominated responses form
a nested hierarchy, which is to say that for $k_{1}\leq k_{2}$, we
have $\mathcal{UD}{}_{i}^{k_{1}}\supseteq\mathcal{UD}{}_{i}^{k_{2}}$.
Further, the union of the dominated and undominated response sets
at any level of the hierarchy is equal to the undominated response
set at the preceding level of the hierarchy, i.e. $\mathcal{D}{}_{i}^{k}\cup\mathcal{UD}{}_{i}^{k}=\mathcal{UD}{}_{i}^{k-1}$. 

A salient point worth emphasizing is that the use of belief hierarchy
sets $W_{i}^{k}$ for $k\geq2$, or dominated and undominated $k^{th}$
order responses for $k\geq2$, entails higher order reasoning. An
agent engages in higher order reasoning when she is reasoning about
the reasoning of other market participants ($\dots$about the reasoning
of other market participants). Recall that $W_{i}^{2}$ imposes a
restriction on the set $B_{i}^{2}$, the set of second order beliefs
of agent $i$. Now, agent $i$'s second order belief is a probability
measure on $Y_{i}^{1}=Y_{i}^{0}\times\prod_{j\neq i}B_{j}^{1}$, and
since $B_{j}^{1}$ is the set of first order beliefs of agent $j\neq i$,
$B_{i}^{2}$ encapsulates agent $i$'s reasoning about the reasoning
of $j$. This is second order reasoning. We use the terms \emph{reasoning}
and \emph{belief }interchangeably since it is traditionally assumed
that an agent forms a belief about an object after she has reasoned
about it; in other words, a higher order belief is the outcome of
an agent's higher reasoning about the belief. Thus, higher the value
of $k$ in $W_{i}^{k}$, higher the order of reasoning that agent
$i$ has employed to arrive at her response.

We justified the definition of belief hierarchy set $W_{i}^{k}$ and
corresponding higher order response sets using the criteria of optimality
and naturalness: it seems unnatural that any agent would choose a
suboptimal, dominated response deliberately, or presume that other
agents would choose such responses, especially when the order $k$
is not very high. One could, however, decide on other criteria of
naturalness to delineate alternative belief hierarchies and their
corresponding higher order response sets. Therefore, we would like
to pin down a few key properties of $W_{i}^{k}$ that become essential
in developing the link with arbitrage. We shall work exclusively with
$W_{i}^{k}$ in the rest of the paper but any belief hierarchy set
that satisfies analogous properties should give us similar results. 

The first of these properties is that there is a whittling down of
the set of belief hierarchies that agent $i$ attributes to agent
$j$ as $i$ uses belief hierarchy sets of higher and higher orders.
That is to say, the belief hierarchies of agent $j$, $j\neq i$,
embedded within $b_{i}\in W_{i}^{k}$ form a nested series of subsets
as the order $k$ in $W_{i}^{k}$ goes up.
\begin{prop}
\label{prop:embedding belief-hierarchy} Denote $b_{j}(W_{i}^{k})=\{b_{j}\in B_{j}:b_{j}\in\supp{\marg_{B_{j}}{[\phi_{i}(b_{i})]}},\,b_{i}\in W_{i}^{k}\}$
for $j\neq i$. Then
\begin{equation}
B_{j}\supseteq b_{j}(W_{i}^{2})\supseteq\dots\supseteq b_{j}(W_{i}^{k-1})\supseteq b_{j}(W_{i}^{k})\supseteq\dots\label{eq:b^Wk set relation}
\end{equation}

That is, a belief hierarchy in $W_{i}^{k}$ embeds a set of belief
hierarchies for agent $j$, given by $b_{j}(W_{i}^{k})$, that is
a subset of the set of hierarchies for $j$ embedded in $W_{i}^{k-1}$,
given by $b_{j}(W_{i}^{k-1})$. 
\end{prop}
\begin{proof}
In the Appendix.
\end{proof}
\smallskip

Proposition \ref{prop:embedding belief-hierarchy} describes how the
belief hierarchy sets that $i$ attributes to other market participants
shrinks as she climbs up her own hierarchy. The next proposition describes
how the \emph{space of strategies} that $i$ considers viable for
other market participants shrinks as she uses belief hierarchies of
higher and higher orders. Just like the belief hierarchies of agent
$j$ embedded within $b_{i}$, the space of strategies for $j$ that
$i$ believes plausible, too, form a nested series of subsets as the
order $k$ in $W_{i}^{k}$ goes up.
\begin{prop}
\label{prop:Responsive actions}Denote $A_{j}(W_{i}^{k})=\{a_{j}\in A_{j}:a_{j}\in\supp{\marg_{A_{j}}{[\phi_{i}(b_{i})]}}$,
$b_{i}\in W_{i}^{k}\}$. Then
\begin{equation}
A_{j}\supseteq A_{j}(W_{i}^{2})\supseteq\dots\supseteq A_{j}(W_{i}^{k-1})\supseteq A_{j}(W_{i}^{k})\dots\label{eq:A^Wk set relation}
\end{equation}

That is, a belief hierarchy in $W_{i}^{k}$ embeds a set of strategies
for agent $j$, given by $A_{j}(W_{i}^{k})$, that is a subset of
the set of strategies for $j$ embedded in $W_{i}^{k-1}$, given by
$A_{j}(W_{i}^{k-1})$. 
\end{prop}
\begin{proof}
In the Appendix.
\end{proof}
\smallskip

The set $A_{j}(W_{i}^{k})$ represents the set of strategies that
agent $i$ deems plausible for agent $j$ when she uses a belief hierarchy
in the set $W_{i}^{k}$, and the definition of $W_{i}^{k}$ implies
that $A_{j}(W_{i}^{k})$ is the set of undominated responses for $j$,
$\mathcal{UD}_{j}^{k-1}$. Thus, intuitively, Proposition \ref{prop:Responsive actions}
is another expression of the nestedness of undominated responses.
Taken in conjunction, Propositions \ref{prop:embedding belief-hierarchy}
and \ref{prop:Responsive actions} indicate that as agent $i$ climbs
higher and higher up the order $k$ in $W_{i}^{k}$, she believes
all market participants use strategies in correspondingly high order
undominated response sets by responding optimally to their respective
belief hierarchies. Since $i\in I$ is a generic agent in the market,
when $k$ is unbounded this is in essence an affirmation of rationality
and common knowledge of rationality in our setting. 

For the sequel, we will assume that equations (\ref{eq:set V_i^2})
\textendash{} (\ref{eq:restriction_general hierarchy to action})
that determine the belief hierarchies in $W_{i}^{k}$ and Definitions
(\ref{def:SD level-1}) and (\ref{def:SD level k}) that determine
dominated and undominated responses on the hierarchies in $W_{i}^{k}$
characterize our agent $i$. Their properties are summarized in Table
\vref{tab:Table Undominated Dominated Belief hierarchy}. We shall
at times refer to the set of belief hierarchies in $W_{i}^{k}$ as
belief hierarchies of order $k$. As noted earlier, as the order $k$
in $W_{i}^{k}$ goes up, so does the order of reasoning used by agent
$i$.

The discussion in this section focused primarily on outlining the
appropriate belief hierarchy for a Bayesian agent in the market. It
turns out that there is an intimate connection between an arbitrage
opportunity and optimal behavior using a belief hierarchy. This is
the subject of the next section.

\section{Arbitrage and Higher Order Beliefs \label{sec:No-arbitrage-and-Higher-Beliefs}}

Having laid out the requisite background in the previous sections,
we are now ready to forge the link between arbitrage and market participant
belief hierarchies that we've been building towards. Specifically,
we see that an arbitrage opportunity results only when an agent underestimates
the degree to which market participants are optimizing: when choosing
her belief hierarchy the agent presumes that market participants will
employ undominated responses of order $k$, but on seeing the actual
asset prices she infers that they have used undominated of responses
of a higher order (Theorems \ref{thm:From Arb to BH} and \ref{thm:From BH to Arb}).
Crucially, determining the presence of an arbitrage opportunity (or
absence) requires the invocation of higher order beliefs of market
participants, and no-abitrage implies that all market participants
are reasoning about other market participants' strategic choices up
to a sufficiently high order (Theorem \ref{thm:No-arbitrage necc =000026 suff}). 

Precisely how many orders of belief need to be invoked depends on
how well the market aggregation mapping $\tilde{f}$ separates among
different sets of market participant choices. This is because the
arbitrage trade relies on a disparity between the initial presumption
of the market agent and what she finds in actuality when she backs
out the real choices made by market participants, using $\tilde{f}$
and the asset prices. If the aggregation mapping $\tilde{f}$ is not
sufficiently responsive, an agent cannot distinguish the actual choices
employed by market participants well, which means that she need not
have reasoned too far up her hierarchy, initially, to have avoided
arbitrage (Propositions \ref{prop:f is one-one} and \ref{prop:f responsiveness}).

An important question raised by this analysis is the extent to which
participants have to \emph{deliberately} reason in market settings
to reach a state of no-arbitrage. Common experience seems to suggest
that real-world traders don't go very far up their belief hierarchies
when reasoning deliberately; yet arbitrage opportunities are hard
to come by in actual markets. Section \ref{subsec:Tatonnement Reasoning}
tackles this question by showing that a market tatonnement process
\textemdash{} in which agents are simply responding to immediate market
circumstances \textemdash{} delivers the same outcomes as the hierarchy-based
reasoning process (Corollary \ref{cor:Corollary tatonnement} and
Proposition \ref{prop:tatonnement and reasoning}). In other words,
the higher order reasoning that we impute to market agents need not
be a completely deliberate process. At each stage, market agents could
be reasoning just one-step ahead \textemdash{} yet stacking up a series
of such ``one-step-aheads'' in a tatonnement sequence would lead
to market outcomes that are indistinguishable from a setting where
agents use higher order reasoning.

\subsection{Link Between Arbitrage and Belief Hierarchy\label{subsec:Link-Between-Arbitrage-Hierarchy}}

This subsection establishes the formal link between arbitrage and
belief hierarchies. As before, we shall primarily reason from the
perspective of a particular market participant, agent $i$, though
we shall make the assumption that all the market agents are symmetric
in how they approach the problem of selecting their respective strategies.
Specifically, the sequence of steps that any agent $i\in I$ follows
will be assumed to be: 
\begin{itemize}[noitemsep]
\item  \emph{Step-1}: Agent $i\in I$ reasons about the market $\mathcal{M}_{i}$
(equations \ref{eq:strategic form market} and \ref{eq:reduced form market})
using a belief hierarchy in the set $W_{i}^{k}$ defined by equations
(\ref{eq:set V_i^2})\textendash (\ref{eq:restriction_general hierarchy to action}); 
\item \emph{Step-2}: She then makes her choice $a_{i}\in A_{i}$, and is
told the market's aggregate selection in the form of an SDF $\tilde{m}$. 
\end{itemize}
We shall term Steps 1 and 2 the \emph{eductive sequence} \label{wrd: eductive sequence}
in the market.\footnote{The term ``eductive'' was introduced in \citet{key-28} to distinguish
play that arises from agents reasoning about the reasoning of other
agents in a game.} We have analyzed Step-1 separately in Section \ref{sec:Hierarchies},
and Step-2 separately in Section \ref{sec:Transforming Arbitrage},
and we now intend to bring the analyses in the two sections together. 

As a preliminary, in order to link Steps 1 and 2 in the eductive sequence,
we shall assume that any market participant's choice of strategy in
Step-2 is undominated with respect to the belief hierarchy she employs
in Step-1. 
\begin{assumption}
\label{assu:Eductive optimize}Each market participant $i\in I$ uses
a belief hierarchy in $W_{i}^{k}$, $k\geq0$, in Step-1 of the eductive
sequence, and their choice of strategy in Step-2 is an undominated
$k^{th}$ order response from $\mathcal{UD}_{i}^{k}$.
\end{assumption}
\smallskip

Assumption \ref{assu:Eductive optimize} simply says that agent $i$
judges a certain subset of strategies plausible for other market participants
in Step-1 of the eductive sequence, and in Step-2 chooses her strategy
to be an undominated response to that subset. Since a belief hierarchy
in $W_{i}^{k}$ attributes plausibility to the subset $\prod_{j\neq i}\mathcal{UD}_{j}^{k-1}\subseteq\prod_{j\neq i}A_{j}$
of strategies for other market participants, another way to state
Assumption \ref{assu:Eductive optimize} is that it mandates that
agent $i$ choose an undominated response to the strategies in $\prod_{j\neq i}\mathcal{UD}_{j}^{k-1}$
in Step-2.

Intuitively, $\prod_{j\neq i}\mathcal{UD}_{j}^{k-1}$ represents agent
$i$'s initial assessment of market participant strategies \emph{before}
she has had a chance to interact with the market. That is to say,
when reasoning with a belief hierarchy in $W_{i}^{k}$ in Step-1 of
the eductive sequence, agent $i$ assigns plausibility to the subset
of strategies $\prod_{j\neq i}A_{j}(W_{i}^{k})=\prod_{j\neq i}\mathcal{UD}_{j}^{k-1}$
for other market participants.\footnote{The set $A_{j}(W_{i}^{k})$ was defined in Proposition \ref{prop:Responsive actions}.}
Subsequently, in Step-2, from the market SDF $\tilde{m}$ and aggregation
mapping $\tilde{f}$, she can back out the set of strategies $A_{-i}(\tilde{m},a_{i})$
that have been actually used by other market participants. $A_{-i}(\tilde{m},a_{i})$
thus represents agent $i$'s assessment of the set of strategies that
were \emph{actually} used by the other market participants in Step-1.
In simple terms, arbitrage arises when the agent's actual finding
differs from her initial assessment. At a technical level, it is the
set-theoretic relationship between $\prod_{j\neq i}\mathcal{UD}_{j}^{k-1}$
and $A_{-i}(\tilde{m},a_{i})$ that determines the nuances of the
link between belief hierarchy and arbitrage, and we explore the connection
in detail in the propositions that follow. 

While we shall be working exclusively with the belief hierarchy set
$W_{i}^{k}$ in this section, most of the results that follow hold
more generally for any set of belief hierarchies that satisfy conditions
analogous to the ones in Propositions \ref{prop:embedding belief-hierarchy}
and \ref{prop:Responsive actions}. Essentially, what we need is that
the space of strategies that agent $i$ deems plausible for $j\neq i$
form a series of subsets as agent $i$ increases the order of her
belief hierarchy set. If this happens, choosing an optimal strategy
using level $k$ automatically implies the strategy is optimal for
levels $0\leq n<k$. Most of the proofs below employ some version
of this argument.

Our first theorem provides a necessary condition for arbitrage. It
says that a tradeable arbitrage opportunity entails at least one market
agent misanticipating the responses of other market participants.
How might such misanticipation arise? Recall that undominated response
sets form a nested hierarchy, i.e. for $k_{1}\leq k_{2}$, we have
$\prod_{j\neq i}\mathcal{UD}{}_{j}^{k_{1}}\supseteq\prod_{j\neq i}\mathcal{UD}{}_{j}^{k_{2}}$.
So, if agent $i$ uses a belief hierarchy in $W_{i}^{k}$ to generate
her own response, anticipating other market participants to select
their responses from their respective undominated $(k-1)^{th}$ order
response sets, the anticipation can go wrong only if other market
participants actually use responses in higher order undominated sets.
That is to say, the agent initially underestimates how far up their
hierarchy other market participants are optimizing (i.e., choosing
undominated responses) and this renders her own initial response suboptimal,
creating an opportunity for a riskless profitable trade when she gets
to know the actual responses of the market participants.
\begin{thm}
\label{thm:From Arb to BH} There is a tradeable arbitrage opportunity
in the market with respect to a probability measure $\mathbf{P}$
only if there is a market participant $i\in I$ who uses a belief
hierarchy in $W_{i}^{k}$ in Step-1 of the eductive sequence and finds
$A_{-i}(\tilde{m},a_{i})\subset\prod_{j\neq i}\mathcal{UD}_{j}^{k-1}$
in Step-2.
\end{thm}
\begin{proof}
\textsf{\small{}From Proposition \ref{prop:Amenable-Arb-def}, a tradeable
arbitrage opportunity means we are given that there is an agent $i\in I$
who may change her strategy from $a_{i}$ to $a_{i}^{*}$ to obtain
\begin{equation}
\tilde{U}_{i}(a_{i}^{*},a_{-i})\geq\tilde{U}_{i}(a_{i},a_{-i})\,\mathbf{P}\text{-a.s.}\,\,\text{and}\,\,\mathbf{P}[\tilde{U}_{i}(a_{i}^{*},a_{-i})>\tilde{U}_{i}(a_{i},a_{-i})]>0,\label{eq:Amenable util arb relation}
\end{equation}
for any $a_{-i}\in A_{-i}(\tilde{m},a_{i})$. We will use a proof
by contradiction to obtain the result. Suppose $A_{-i}(\tilde{m},a_{i})\supseteq\prod_{j\neq i}\mathcal{UD}_{j}^{k-1}$.
Recall that for undominated response sets, $\mathcal{UD}{}_{i}^{m1}\supseteq\mathcal{UD}{}_{i}^{m2}$
when $m1\leq m2$ (see Definition \ref{def:SD level k}), i.e. the
undominated response sets form a nested sequence as one increases
the order. By Assumption \ref{assu:Eductive optimize}, $a_{i}\in\mathcal{UD}_{i}^{k}$,
which means that for any strategy $a_{-i}$ that other market participants
choose in the set $\prod_{j\neq i}\mathcal{UD}_{j}^{k-1}$, agent
$i$ cannot unequivocally improve her utility}\footnote{\textsf{\footnotesize{}Recall that we }\textsf{\small{}t}\textsf{\footnotesize{}he
term ``unequivocally improve utility'' is a short-form for condition
(\ref{eq:Amenable util arb relation}), i.e. a change of strategy
from $a_{i}$ to $a_{i}^{*}$ for agent $i$ leading to
\[
\tilde{U}_{i}(a_{i}^{*},a_{-i})\geq\tilde{U}_{i}(a_{i},a_{-i})\,\mathbf{P}\text{-a.s.}\,\,\text{and}\,\,\mathbf{P}[\tilde{U}_{i}(a_{i}^{*},a_{-i})>\tilde{U}_{i}(a_{i},a_{-i})]>0\,\,\,\text{for any \ensuremath{a_{-i}\in A_{-i}}(\ensuremath{\tilde{m}},\ensuremath{a_{i}})}.
\]
}}\textsf{\small{} by changing her strategy. Therefore, when $A_{-i}(\tilde{m},a_{i})\supseteq\prod_{j\neq i}\mathcal{UD}_{j}^{k-1}$,
agent $i$ cannot unequivocally improve her utility by changing her
strategy if she is already choosing her strategy from the set $\mathcal{UD}_{i}^{k}$.
This means there is no $a_{i}^{*}$ satisfying condition (\ref{eq:Amenable util arb relation})
above. Therefore, we have the result.}{\small\par}
\end{proof}
\medskip

Assumption \ref{assu:Eductive optimize} links the agent's belief
hierarchy to her undominated response, so an equivalent way to state
Proposition \ref{thm:From Arb to BH} is that there is a tradeable
arbitrage opportunity in the market only if there is a market participant
$i\in I$ who uses a response in $\mathcal{UD}_{i}^{k}$ but finds
$A_{-i}(\tilde{m},a_{i})\subset\prod_{j\neq i}\mathcal{UD}_{j}^{k-1}$.

Is the condition in Theorem \ref{thm:From Arb to BH} also sufficient
for arbitrage? Not quite. To see why, notice that a belief hierarchy
in $W_{i}^{k}$ \textemdash{} which requires that agent $i$ choose
a strategy in the undominated set $\mathcal{UD}_{i}^{k}$ under Assumption
\ref{assu:Eductive optimize} \textemdash{} doesn't completely pin
down agent $i$'s strategy set. Since undominated responses form nested
subsets, agent $i$ could very well be choosing a strategy in the
set $\bigcap_{k\geq0}\mathcal{UD}_{i}^{k}$ when the only restriction
is that she select in $\mathcal{UD}_{i}^{k}$. If this is the case,
however, she has no tradeable arbitrage opportunity no matter what
responses other market participants select. This is because agent
$i$ has already chosen from the best possible response set she can
find \textemdash{} a response in $\bigcap_{k\geq0}\mathcal{UD}_{i}^{k}$
is equivalent to selecting an undominated response using a belief
hierarchy in $\bigcap_{k\geq0}W_{i}^{k}$ \textemdash{} and there
is no way for her to improve it any further. Intuitively, she has
exhausted all the infinite orders of her hierarchy and, therefore,
her response cannot be made any better. We need to exclude this case
from our specification of undominated responses to derive a sufficient
condition. In other words, we need to specify a set $\mathcal{UD}_{i}^{k}\setminus\mathcal{UD}_{i}^{k+n}$,
$n$ finite, for agent $i$'s strategy, in addition to specifying
the belief hierarchy set $W_{i}^{k}$. Since $\mathcal{UD}_{i}^{k}\setminus\mathcal{UD}_{i}^{k+n}=\bigcup_{l=0}^{n}\mathcal{UD}_{i}^{k+l}\setminus\mathcal{UD}_{i}^{k+l+1}$,
we can as well work with the set $\mathcal{UD}_{i}^{k}\setminus\mathcal{UD}_{i}^{k+1}$
without loss of generality, and this is what we do.

The condition $A_{-i}(\tilde{m},a_{i})\subset\prod_{j\neq i}\mathcal{UD}_{j}^{k-1}$
in Theorem \ref{thm:From Arb to BH}, too, needs a slight modification
to generate sufficiency. If agent $i$ selects an undominated response
to the set $\prod_{j\neq i}\mathcal{UD}_{j}^{k-1}$, by the nestedness
of undominated sets, the response is automatically undominated for
the sets $\prod_{j\neq i}\mathcal{UD}_{j}^{k-2}$, $\prod_{j\neq i}\mathcal{UD}_{j}^{k-3}$,
$\dots$, $\prod_{j\neq i}\mathcal{UD}_{j}^{0}$. This is the rationale
behind the requirement that $A_{-i}(\tilde{m},a_{i})$ be a strict
subset of $\prod_{j\neq i}\mathcal{UD}_{j}^{k-1}$ when there is an
arbitrage opportunity. However, $A_{-i}(\tilde{m},a_{i})\subset\prod_{j\neq i}\mathcal{UD}_{j}^{k-1}$
is a characterization based on specifying the regions where arbitrage
\emph{cannot} exist, it does not say anything about the regions where
the arbitrage opportunity \emph{does} exist. For sufficiency, we need
to recognize that $\mathcal{UD}_{i}^{k}\setminus\mathcal{UD}_{i}^{k+1}=\mathcal{D}_{i}^{k+1}$
(see Figure \vref{fig:Domundom}), which is a dominated response to
the set $\prod_{j\neq i}\mathcal{UD}_{j}^{k}$. Now, a dominated response
to a set implies the reponse is also dominated with respect to every
subset of the said set. Therefore, whenever $A_{-i}(\tilde{m},a_{i})$
belongs in a subset of $\prod_{j\neq i}\mathcal{UD}_{j}^{k}$ the
agent has an arbitrage opportunity. These sufficient conditions are
summarized in Theorem \ref{thm:From BH to Arb}. 
\begin{thm}
\label{thm:From BH to Arb} If there is a market participant $i\in I$
who uses a belief hierarchy in $W_{i}^{k}$ in Step-1 of the eductive
sequence, and selects an undominated response in the set $\mathcal{UD}_{i}^{k}\setminus\mathcal{UD}_{i}^{k+1}$
to find $A_{-i}(\tilde{m},a_{i})\subseteq\prod_{j\neq i}\mathcal{UD}_{j}^{k}$
in Step-2, there is a tradeable arbitrage opportunity in the market
with respect to the probability measure $\mathbf{P}$. 
\end{thm}
\begin{proof}
\textsf{\small{}From condition (\ref{eq:undominated W^k}) in Definition
\ref{def:SD level k}, $\mathcal{UD}_{i}^{k}\setminus\mathcal{UD}_{i}^{k+1}=\mathcal{D}_{i}^{k+1}$.
The definition of $\mathcal{D}_{i}^{k+1}$ implies that agent $i$
may change her strategy from $a_{i}$ to $a_{i}^{*}$ to obtain
\begin{equation}
\tilde{U}_{i}(a_{i}^{*},a_{-i})\geq\tilde{U}_{i}(a_{i},a_{-i})\,\mathbf{P}\text{-a.s.}\text{ and}\,\,\mathbf{P}[\tilde{U}_{i}(a_{i}^{*},a_{-i})>\tilde{U}_{i}(a_{i},a_{-i})]>0,\label{eq:dominated for BH to arb}
\end{equation}
for any $a_{-i}\in\prod_{j\neq i}\mathcal{UD}_{j}^{k}$. Next, as
noted before, undominated response sets have the property $\mathcal{UD}{}_{i}^{m1}\supseteq\mathcal{UD}{}_{i}^{m2}$
when $m1\leq m2$, i.e. the undominated response sets form a nested
sequence as one increases the order. Thus, $A_{-i}(\tilde{m},a_{i})\subseteq\prod_{j\neq i}\mathcal{UD}_{j}^{k}$
means condition (\ref{eq:dominated for BH to arb}) above holds for
any $a_{-i}\in A_{-i}(\tilde{m},a_{i})$. From Proposition \ref{prop:Amenable-Arb-def},
this implies that there is a tradeable arbitrage opportunity in the
market with respect to the probability measure $\mathbf{P}$, since
agent $i$ may change her strategy to $a_{i}^{*}$. Therefore, we
have the result.}{\small\par}
\end{proof}
\medskip

In sum, Theorems \ref{thm:From Arb to BH} and \ref{thm:From BH to Arb}
say that a tradeable arbitrage opportunity in the market is equivalent
to there being at least one agent in the market who has not exhausted
her entire order hierarchy when choosing her response, and who finds
(when she sees the actual asset prices) that other market participants
are going further up their respective hierarchies in selecting responses
than she had anticipated. 

The contrapositive of the above theorems gives us a characterization
of no-arbitrage in this setup. As we've already discussed in the context
of Theorem \ref{thm:From BH to Arb}, if market participants are exhausting
their entire order hierarchy when choosing their response (i.e., choosing
in the set $\bigcap_{n\geq0}\mathcal{UD}_{i}^{n}$), there is no way
for them to improve their responses any further. This is one of the
avenues through which no-arbitrage may be achieved in markets. The
second avenue is through agents choosing an undominated response with
order just high enough that they cannot improve on it even after observing
the actual asset prices. It is this second avenue for no-arbitrage
that makes the market environment distinct from traditional game-theoretic
environments (more on this below). Theorem \ref{thm:No-arbitrage necc =000026 suff}
shows that these are the only two avenues for no-arbitrage in markets.
Like the characterizations in the fundamental theorems of asset pricing
(\citealt{key-18}), the theorem generates an alternative definition
for arbitrage-free markets. 
\begin{thm}
\label{thm:No-arbitrage necc =000026 suff} \ 
\begin{enumerate}[noitemsep]
\item[i.]  There is no tradeable arbitrage opportunity in the market \uline{if}
every market participant $i\in I$ either chooses a response in the
set $\mathcal{UD}_{i}^{k_{i}}\setminus\mathcal{UD}_{i}^{k_{i}+1}$
and finds $A_{-i}(\tilde{m},a_{i})\supseteq\prod_{j\neq i}\mathcal{UD}_{j}^{k_{i}-1}$,
or chooses a response in the set $\bigcap_{n\geq0}\mathcal{UD}_{i}^{n}$.
\item[ii.] There is no tradeable arbitrage opportunity in the market \uline{only
if} every market participant $i\in I$ either chooses a response in
the set $\mathcal{UD}_{i}^{k_{i}}\setminus\mathcal{UD}_{i}^{k_{i}+1}$
and finds $A_{-i}(\tilde{m},a_{i})\supset\prod_{j\neq i}\mathcal{UD}_{j}^{k_{i}}$,
or chooses a response in the set $\bigcap_{n\geq0}\mathcal{UD}_{i}^{n}$.
\end{enumerate}
\end{thm}
\begin{proof}
\textsf{\small{}Part-1a. If every market participant $i\in I$ chooses
a response in the set $\mathcal{UD}_{i}^{k}\setminus\mathcal{UD}_{i}^{k+1}$
and finds $A_{-i}(\tilde{m},a_{i})\supseteq\prod_{j\neq i}\mathcal{UD}_{j}^{k_{i}-1}$
there is no tradeable arbitrage opportunity:}{\small\par}

\textsf{\small{}If a market participant $i\in I$ chooses a response
in the set $\mathcal{UD}_{i}^{k_{i}}\setminus\mathcal{UD}_{i}^{k_{i}+1}$,
we have that $a_{i}\in\mathcal{UD}_{i}^{k_{i}}$. Thus, $a_{i}$ is
an undominated response for any $a_{-i}\in\prod_{j\neq i}\mathcal{UD}_{j}^{k_{i}-1}$.
Further, since $A_{-i}(\tilde{m},a_{i})\supseteq\prod_{j\neq i}\mathcal{UD}_{j}^{k_{i}-1}$,
$a_{i}$ is undominated for at least some $a_{-i}\in A_{-i}(\tilde{m},a_{i})$.
Therefore, there is no $a_{i}^{*}$ such that $i$ may change her
strategy from $a_{i}$ to $a_{i}^{*}$ to obtain
\begin{equation}
\tilde{U}_{i}(a_{i}^{*},a_{-i})\geq\tilde{U}_{i}(a_{i},a_{-i})\,\mathbf{P}\text{-a.s.}\,\,\text{and}\,\,\mathbf{P}[\tilde{U}_{i}(a_{i}^{*},a_{-i})>\tilde{U}_{i}(a_{i},a_{-i})]>0,\label{eq:Again Amenable util arb}
\end{equation}
for every $a_{-i}\in A_{-i}(\tilde{m},a_{i})$. If this condition
is true for all market participants, the contrapositive of Corollary
\ref{cor:Corollary Amenable Arb} implies there is no tradeable arbitrage
opportunity in the market.}{\small\par}

\textsf{\small{}Part-1b. If every market participant chooses a response
in the set $\bigcap_{n\geq0}\mathcal{UD}_{i}^{n}$ there is no tradeable
arbitrage opportunity:}{\small\par}

\textsf{\small{}If a market participant $i\in I$ chooses a response
in the set $\bigcap_{n\geq0}\mathcal{UD}_{i}^{n}$, she is in effect
using a belief hierarchy in $\bigcap_{k\geq0}W_{i}^{k}$, and has
thus exhausted all the orders of the belief hierarchy set. Thus, there
is no $a_{i}^{*}$ such that $i$ may change her strategy from $a_{i}$
to $a_{i}^{*}$ to obtain
\begin{equation}
\tilde{U}_{i}(a_{i}^{*},a_{-i})\geq\tilde{U}_{i}(a_{i},a_{-i})\,\mathbf{P}\text{-a.s.}\,\,\text{and}\,\,\mathbf{P}[\tilde{U}_{i}(a_{i}^{*},a_{-i})>\tilde{U}_{i}(a_{i},a_{-i})]>0,\label{eq:Again again Amenable util arb}
\end{equation}
no matter what $a_{-i}$ the other market participants select. In
particular, this holds for $a_{-i}\in A_{-i}(\tilde{m},a_{i})$. If
no participant can choose an $a_{i}^{*}$ satisfying the condition,
the contrapositive of Corollary \ref{cor:Corollary Amenable Arb}
implies there is no tradeable arbitrage opportunity in the market.}{\small\par}

\textsf{\small{}Part 2. There is no tradeable arbitrage opportunity
in the market only if every market participant $i\in I$ either chooses
a response in the set $\mathcal{UD}_{i}^{k_{i}}\setminus\mathcal{UD}_{i}^{k_{i}+1}$
to find $A_{-i}(\tilde{m},a_{i})\supset\prod_{j\neq i}\mathcal{UD}_{j}^{k_{i}},$
or chooses a response in the set $\bigcap_{n\geq0}\mathcal{UD}_{i}^{n}$:}{\small\par}

\textsf{\small{}We will employ a proof by contradiction for this part.
We work with the contrapositive of the statement. Suppose there is
a market participant $i\in I$ who neither chooses a response in the
set }$\bigcap_{n\geq0}\mathcal{UD}_{i}^{n}$ \textsf{\small{}nor chooses
a response in the set $\mathcal{UD}_{i}^{k}\setminus\mathcal{UD}_{i}^{k+1}$
to find $A_{-i}(\tilde{m},a_{i})\supset\prod_{j\neq i}\mathcal{UD}_{j}^{k_{i}}$.
In this case, $i$ will have chosen a response in the set $\mathcal{UD}_{i}^{k}\setminus\mathcal{UD}_{i}^{k+1}$
to find $A_{-i}(\tilde{m},a_{i})\subseteq\prod_{j\neq i}\mathcal{UD}_{j}^{k_{i}}$.}\footnote{\textsf{\footnotesize{}Recall that $\neg(P\longrightarrow Q)$, the
negation of an implication, is $P\land\neg Q$. In our case $P$ is
``chooses a response in the set $\mathcal{UD}_{i}^{k}\setminus\mathcal{UD}_{i}^{k+1}$''
and $Q$ is ``finds $A_{-i}(\tilde{m},a_{i})\supset\prod_{j\neq i}\mathcal{UD}_{j}^{k_{i}}$''.}}\textsf{\small{} By Theorem \ref{thm:From BH to Arb} we then have
that there is a tradeable arbitrage opportunity in the market. Thus,
a contradiction.}{\small\par}
\end{proof}
\medskip

Taken together, Theorems \ref{thm:From Arb to BH}\textendash \ref{thm:No-arbitrage necc =000026 suff}
illustrate that arbitrage in markets could be described as a purely
belief-based phenomenon, without any explicit reference to asset fundamentals.
The most important ingredient is the belief hierarchy, and no-arbitrage
requires that market participants optimize till a sufficiently high
order. In other words, checking the optimality of first order beliefs
is not enough to ensure no-arbitrage in markets, and higher order
beliefs can play a pivotal role. If a market is arbitrage-free, it
is quite likely that agents are engaging in higher order reasoning
using their belief hierarchies.\footnote{As discussed in Section \ref{subsec:Dominated-Responses}, an agent
engages in higher order reasoning when she is reasoning about the
reasoning $\dots$ about the reasoning of other market participants,
and a higher order belief is the outcome of an agent's higher reasoning.} Such higher order reasoning need not always be deliberate, as we
highlight in the next subsection, and could also follow from a mechanical
process of tatonnement. Nevertheless, higher order beliefs and higher
order reasoning are an essential component of arbitrage-free markets
and likely play a much more fundamental role in asset pricing than
traditionally envisaged for them. 

At this point, it may be worthwhile to note the distinctions between
no-arbitrage and epistemic game-theoretic solution concepts like rationalizability
(\citealt{key-5}, \citealt{key-3}) or rationality and common knowledge
of rationality (\citealt{key-21}, \citealt{key-20}). First, of course,
there is the difference in the optimization criteria discussed in
Section \ref{subsec:Dominated-Responses}. More importantly, the epistemic
notions mandate that agents must exercise all the infinite orders
of their hierarchy, i.e. the agents have to select from (the equivalent
of) $\bigcap_{n\geq0}\mathcal{UD}_{i}^{n}$, and there is no provision
for taking in feedback from the game using a set like $A_{-i}(\tilde{m},a_{i})$.
In case of no-arbitrage, on the other hand, all agents $i\in I$ selecting
from $\bigcap_{n\geq0}\mathcal{UD}_{i}^{n}$ is only one avenue to
achieve the outcome. No-arbitrage may equally well be achieved through
agents choosing an undominated response with order just high enough
that they cannot improve on it even after observing the actual asset
prices.

Whether or not $\prod_{j\neq i}\mathcal{UD}_{j}^{k_{i}-1}$ is a subset
of $A_{-i}(\tilde{m},a_{i})$ depends on how well the aggregation
mapping $\tilde{f}$ distinguishes among sets of market participant
choices. For example, if the stochastic discount factor $\tilde{m}$
stays unchanged no matter what strategies $(a_{i})_{i\in I}$ market
participants choose, we have $A_{-i}(\tilde{m},a_{i})=\prod_{j\neq i}A_{j}$.
This means that $A_{-i}(\tilde{m},a_{i})\supseteq\prod_{j\neq i}\mathcal{UD}_{j}^{k_{i}-1}$
no matter which order value $k_{i}\geq1$ is used.\footnote{Recall that $\mathcal{UD}_{j}^{0}=A_{j}$ by convention, so that for
$k=1$ we have $\prod_{j\neq i}A_{j}=\prod_{j\neq i}\mathcal{UD}_{j}^{k-1}$. } That is to say, the no-arbitrage condition places no restriction,
whatsoever, on the belief hierarchies that may prevail in the market
in this case. In such a market, there is no possibility of arbitrage
regardless of which belief hierarchies are used because market participants
cannot distinguish at all the beliefs that support the observed prices. 

On the other hand, if $\tilde{m}$ changes for every new combination
of market participant choices (i.e., $\tilde{f}$ is one-to-one),
then Proposition \ref{prop:f is one-one} shows that no-arbitrage
implies that all market participants are exercising their entire infinite
hierarchy of beliefs.
\begin{prop}
\label{prop:f is one-one}When $\tilde{f}$ is one-to-one, there is
no tradeable arbitrage opportunity in the market if and only if every
market participant $i\in I$ chooses a response in the set $\bigcap_{n\geq0}\mathcal{UD}_{i}^{n}$.
\end{prop}
\begin{proof}
In the Appendix.
\end{proof}
\smallskip

In other words, if the aggregation mapping $\tilde{f}$ allows market
participants to distinguish the beliefs supporting asset prices really
well, then the participants must reason about rather high order beliefs
if they have to achieve a state of no-arbitrage. Proposition \ref{prop:f is one-one}
encapsulates the extreme end of this story \textemdash{} with $\tilde{f}$
one-to-one, agents must choose in the set $\bigcap_{n\geq0}\mathcal{UD}_{i}^{n}$.
In fact, we can derive a comparative statics result based on the responsiveness
of the aggregation mapping that highlights this narrative more generally.

Recall from (\ref{eq:A_i(m,a_i)}) that $a_{i}\times A_{-i}(\tilde{m},a_{i})=\tilde{f}^{-1}(\tilde{m})$.
The notation $A_{-i}(\tilde{m},a_{i})$ presumes that the aggregation
mapping $\tilde{f}$ is fixed. When considering comparative statics
involving the market aggregation mapping, it helps to use the notation
$A_{-i}^{\tilde{f}}(\tilde{m},a_{i})$ to indicate that the inversion
is undertaken with respect to the mapping $\tilde{f}$. We say that
market aggregation mapping $\tilde{f}_{2}$ is more \emph{responsive}
than market aggregation mapping $\tilde{f}_{1}$ if 
\begin{equation}
A_{-i}^{\tilde{f_{2}}}(\tilde{m},a_{i})\subseteq A_{-i}^{\tilde{f_{1}}}(\tilde{m},a_{i})\,\,\,\forall i\in I.\label{eq:responsiveness condition}
\end{equation}
In intuitive terms, $\tilde{f}_{2}$ allows agents to distinguish
among market participant choices better than $\tilde{f}_{1}$. If
this is the case, the proposition below shows that agents must be
using a belief hierarchy set with higher order under $\tilde{f}_{2}$,
than under $\tilde{f}_{1}$, to attain no-arbitrage.
\begin{prop}
\label{prop:f responsiveness} If $\tilde{f}_{2}$ is more responsive
than $\tilde{f}_{1}$ \textemdash{} and $k_{i2}$, $k_{i1}$ are the
minimum orders for which $A_{-i}^{\tilde{f_{2}}}(\tilde{m},a_{i})\supset\prod_{j\neq i}\mathcal{UD}_{j}^{k_{i2}}$
and $A_{-i}^{\tilde{f_{1}}}(\tilde{m},a_{i})\supset\prod_{j\neq i}\mathcal{UD}_{j}^{k_{i1}}$,
$i\in I$ \textemdash{} then
\begin{equation}
k_{i2}\geq k_{i1},
\end{equation}
and the minimum order of the undominated response set from which a
market participant selects, given no-arbitrage, increases weakly under
$\tilde{f}_{2}$ (ceteris paribus).
\end{prop}
\begin{proof}
In the Appendix.
\end{proof}
\smallskip

In order to get the intuitive gist of Propositions \ref{prop:f is one-one}
and \ref{prop:f responsiveness}, it helps to go over a somewhat heuristic
analogy. Imagine a trader in the market deciding how many orders of
reasoning to use. In order to leave no money on the table, the trader
knows he has to reason one step ahead of the market (more accurately,
ensure $\prod_{j\neq i}\mathcal{UD}_{j}^{k-1}\subseteq A_{-i}^{\tilde{f}}(\tilde{m},a_{i})$).
Since the objective is to stay one step ahead, how many steps ahead
the trader actually reasons depends on his perception of the number
of steps the market reasons. After he observes the market prices,
this perception of the trader is captured in a variable like $A_{-i}^{\tilde{f}}(\tilde{m},a_{i})$.
When the inverse aggregation mapping is completely precise ($\tilde{f}$
is one-one), so that the trader knows the number of steps used by
the market exactly, there ensues a competitive game \textemdash{}
if our trader is one step ahead of the market, then the market is
one step behind and plays catch-up in order to leave no money on the
table, and vice-versa \textemdash{} that pushes all traders to reason
ahead an unbounded number of steps. On the other hand, if the inverse
aggregation mapping conveys a rather vague description of the number
of steps ($\tilde{f}$ is not very responsive), then each trader can
convince himself that he is leaving no money on the table despite
not reasoning very many steps ahead. This happens because arbitrage
uses a rather demanding interpretation for ``money on the table''
in case of vague descriptions: only when no money is lost in \emph{every}
scenario that can be conceived under the vague description, and some
money gets made in at least some scenarios, do we have an arbitrage.
As the descriptions get more and more vague, the number of scenarios
that may be conceived under the descriptions multiply, rendering it
ever more likely for traders to uncover a scenario where they can
lose money \textemdash{} despite not reasoning too many steps ahead.

In narrative above, a very important role is played by the reasoning
faculty of agents. Is it possible for markets to be arbitrage-free
when agents do not posess such powers of reason? This is the question
we tackle in the next subsection.

\subsection{Tatonnement versus Reasoning \label{subsec:Tatonnement Reasoning}}

In the previous subsection we analyzed no-arbitrage through the lens
of a market participant's reasoning. We could also proceed with a
similar analysis using the lens of market tatonnement. Corollary \ref{cor:Corollary Amenable Arb}
told us that a necessary and sufficient condition for a tradeable
arbitrage opportunity is the presence of a market participant who
uses a dominated-wrtp response. Another way to state the corollary
would be as follows.
\begin{cor}
\emph{(Proposition \ref{prop:Amenable-Arb-def})}.\label{cor:Corollary tatonnement}
There is no tradeable arbitrage opportunity in the market with respect
to a probability measure $\mathbf{P}$ if and only if there is no
market participant $i\in I$ who is using a dominated-wrtp response.
\end{cor}
\smallskip

Corollaries \ref{cor:Corollary Amenable Arb} and \ref{cor:Corollary tatonnement}
make no reference to the reasoning process (if any) employed by market
participants, and in this subsection we make no assumptions on how
market agents actually undertake their reasoning (they could very
well be not reasoning at all, at least deliberately). Our sole assumption
here will be that agents trade away any immediate arbitrage opportunity
that becomes available to them. The main takeaway is that this single
assumption about trading behavior results in a market adjustment process
that is, for all purposes, equivalent to\emph{ }market participants
reasoning about progressively higher orders of their belief hierarchy.
Such an adjustment process continues until a no-arbitrage condition
like Corollary \ref{cor:Corollary tatonnement} (or Theorem \ref{thm:No-arbitrage necc =000026 suff},
when employing belief hierarchies) is met. In other words, even when
agents are not deliberately using higher order beliefs, \emph{they
end up behaving as if they do, }as long as they trade away arbitrage
opportunities.

We use the term\emph{ tatonnement} to describe the adjustment process
by which a market with arbitrage opportunities transforms into a market
with no-arbitrage. More precisely, market tatonnement describes the
following algorithm:
\begin{itemize}[noitemsep]
\item  \emph{Step-1}: Each market participant $i\in I$ selects a strategy
$a_{i}\in A_{i}$ and then gets to know the market SDF $\tilde{m}_{1}=\tilde{f}((a_{i})_{i\in I})$.
If no market participant finds that she is using a dominated-wrtp
response, the process concludes. Otherwise, the process moves to Step-$n$,
with $n=2$.
\item \emph{Step-}$n$: A market participant $i\in I$ who finds that she
has used a strategy in her dominated-wrtp response set $D_{i}^{wrtp}$,
given $a_{i}$, $\tilde{m}_{n-1}$ and $\tilde{f}$, selects an alternative
strategy $a_{i}^{*}$ from her undominated-wrtp response set $UD_{i}^{wrtp}$.
This gives rise to a new market SDF $\tilde{m}_{n}$. After these
adjustments, if no market participant finds that she is using a dominated-wrtp
response, the process concludes. Otherwise, a counter $l$ is first
set to $l=n$, and the process moves back to Step-$n$, but with $n=l+1$. 
\end{itemize}
Thus, the tatonnement process concludes when the condition in Corollary
\ref{cor:Corollary tatonnement} is satisfied and the market is arbitrage-free.

At the outset, let us note that a market where the tatonnement process
has concluded is equivalent (in terms of outcome) to a market in which
every market participant $i\in I$ either chooses a response in the
set $\mathcal{UD}_{i}^{k_{i}}\setminus\mathcal{UD}_{i}^{k_{i}+1}$
and finds $A_{-i}(\tilde{m},a_{i})\supset\prod_{j\neq i}\mathcal{UD}_{j}^{k_{i}}$,
or chooses a response in the set $\bigcap_{n\geq0}\mathcal{UD}_{i}^{n}$.
This follows from the fact that both Corollary \ref{cor:Corollary tatonnement}
and Theorem \ref{thm:No-arbitrage necc =000026 suff} are equivalent
descriptions of an arbitrage-free market. An outside analyst who has
access to only the final outcome of the market adjustment process
(i.e. the final asset prices and final quantities traded by market
participants) has no way of guessing whether the tatonnement algorithm,
or the belief hierarchy based reasoning procedure in Section \ref{subsec:Link-Between-Arbitrage-Hierarchy},
was used in arriving at the outcome. 

The proposition below shows that a similar outcome equivalence holds
for the most part at every step of the tatonnement algorithm.
\begin{prop}
\label{prop:tatonnement and reasoning}A market participant $i\in I$
changing her strategy from $a_{i}\in D_{i}^{wrtp}$ initially, to
$a_{i}^{*}\in UD_{i}^{wrtp}$ subsequently, corresponds to the following
change,
\begin{enumerate}[noitemsep]
\item[(i)]  Initially, the participant uses a belief hierarchy in $W_{i}^{k}$
for the eductive sequence and finds $A_{-i}(\tilde{m},a_{i})\subset\prod_{j\neq i}\mathcal{UD}_{j}^{k-1}$,
\item[(ii)] Subsequently, the participant either uses a belief hierarchy in $W_{i}^{k+\alpha}$
for the eductive sequence and finds $A_{-i}(\tilde{m},a_{i})\supset\prod_{j\neq i}\mathcal{UD}_{j}^{k+\alpha}$,
or uses a belief hierarchy in $\bigcap_{\beta>\alpha}W_{i}^{k+\beta}$,
for some integer $\alpha\geq0$.
\end{enumerate}
\end{prop}
\begin{proof}
In the Appendix.
\end{proof}
\medskip

Proposition \ref{prop:tatonnement and reasoning} implies that each
step of the tatonnement process increases the order of the belief
hierarchy set for some agent $i\in I$ by the quantity $\alpha$ (i.e.,
$W_{i}^{k}$ to $W_{i}^{k+\alpha}$). The value of $\alpha$ may be
zero because of the gap between necessary and sufficient conditions
for arbitrage,\footnote{More discussion on the $\alpha=0$ scenario follows below, see footnote
\ref{fn:Alpha_Zero}.} but $\alpha$ is always non-negative. In other words, every round
of tatonnement (weakly) increases the order of the belief hierarchy
of a market participant who trades away the arbitrage opportunity
in that round. Thus, as more and more rounds of the tatonnement process
pile on, it gets increasingly likely that agents in the market are
all exercising their higher order reasoning \textemdash{} assuming
the market is arbitrage-free. This climb higher and higher up the
belief hierarchy is especially striking because it does not need the
agents to ``deliberately'' engage in higher order reasoning. In
each round of tatonnement, the agents are simply trading away an immediate
arbitrage opportunity. Nevertheless, the cumulative effect of repeated
rounds of tatonnement is \emph{as if} the agents are undertaking higher
order reasoning.

To see the content of Proposition \ref{prop:tatonnement and reasoning}
more intuitively, it helps to note that one may partition any individual
market participant's strategy space as follows:\textsf{
\begin{equation}
A_{i}=\bigcup_{n\geq0}\mathcal{D}_{i}^{n}\,\,\bigcup\,\,\left(\bigcap_{l\geq0}\mathcal{UD}_{i}^{l}\right),\,\,\,i\in I.\label{eq:partition D_k}
\end{equation}
}\begin{figure}[t]
\begin{tikzpicture}[domain=0:4][font=\small]  
\draw[step=0.1,very thin,color=white!20] (-6,0) grid (10,9); 
\draw[very thick][-] (-1.5,1) -- (6.75,1);
\node[font=\normalsize] [below] at (-1.125,1){$\mathcal{D}_i^1$};
\node[font=\normalsize] [below] at (-0.375,1){$\mathcal{D}_i^2$};
\node[font=\normalsize] [below] at (0.375,1){$\mathcal{D}_i^3$};
\node[font=\normalsize] [below] at (1.125,1){$\mathcal{D}_i^4$};
\node[font=\normalsize] [below] at (1.875,1){$\mathcal{D}_i^5$};
\node[font=\normalsize] [below] at (2.625,1){$\mathcal{D}_i^6$};
\node[font=\normalsize] [below] at (3.375,1){$\mathcal{D}_i^7$};
\node[font=\large] [below] at (4.2,0.7){$\dots$};
\node[font=\normalsize] [below] at (6.2,1){$\bigcap_{l\geq 0}\mathcal{UD}_i^l$};
\draw[thin][color=black][-] (-1.75,1.75) -- (6,1.75);
\draw[thin][color=black][-] (-1.75,2.5) -- (6,2.5);
\draw[thin][color=black][-] (-1.75,3.25) -- (6,3.25);
\draw[thin][color=black][-] (-1.75,4) -- (6,4);
\draw[thin][color=black][-] (-1.75,4.75) -- (6,4.75);
\draw[thin][color=black][-] (-1.75,5.5) -- (6,5.5);
\draw[thin][color=black][-] (-1.75,6.25) -- (6,6.25);
\draw[thin][color=black][-] (-0.75,0.75) -- (-0.75,8);
\draw[thin][color=black][-] (0,0.75) -- (0,8);
\draw[thin][color=black][-] (0.75,0.75) -- (0.75,8);
\draw[thin][color=black][-] (1.5,0.75) -- (1.5,8);
\draw[thin][color=black][-] (2.25,0.75) -- (2.25,8);
\draw[thin][color=black][-] (3,0.75) -- (3,8);
\draw[thin][color=black][-] (3.75,0.75) -- (3.75,8);
\draw[very thick][-] (-1.5,1) -- (-1.5,8);
\node[font=\normalsize] [left] at (-1.5,1.37){$\mathcal{D}_j^1$};
\node[font=\normalsize] [left] at (-1.5,2.12){$\mathcal{D}_j^2$};
\node[font=\normalsize] [left] at (-1.5,2.87){$\mathcal{D}_j^3$};
\node[font=\normalsize] [left] at (-1.5,3.62){$\mathcal{D}_j^4$};
\node[font=\normalsize] [left] at (-1.5,4.37){$\mathcal{D}_j^5$};
\node[font=\normalsize] [left] at (-1.5,5.12){$\mathcal{D}_j^6$};
\node[font=\normalsize] [left] at (-1.5,5.87){$\mathcal{D}_j^7$};
\node[font=\large] [left] at (-1.7,6.7){$\vdots$};
\node[font=\normalsize] [left] at (-1.5,7.8){$\bigcap_{l\geq 0}\mathcal{UD}_j^l$};
\pic[line width=2pt,red] at (1,5.20) {mycross};
\pic[line width=2pt,blue] at (3.5,5.20) {mycross};
\draw[very thick,->] (1.25,5.20) -- (3.25,5.20);
\fill[pattern=horizontal lines, pattern color=gray] (-1.5,4.37) rectangle (7,8.25);
\fill[pattern=vertical lines, pattern color=gray] (1.875,1) rectangle (7,8.25);
\node[font=\footnotesize] [below] at (2.75,9.25){From agent $i$'s perspective all $j$'s responses in the};
\node[font=\footnotesize] [below] at (2.75,8.75){horizontal hatched region belong in the same set};
\node[font=\footnotesize] [below] at (6.25,3){From agent $j$'s perspective all $i$'s responses in the};
\node[font=\footnotesize] [below] at (6.25,2.5){vertical hatched region belong in the same set};
\node[font=\footnotesize] [below] at (2.5,6){Tatonnement adjustment};
\end{tikzpicture}
\caption{ \textsf{\bf Tatonnement and reasoning.} \textsf{Tatonnement adjustments in the market can be represented using participant belief hierarchies and reasoning.}}\label{fig:Tatonnement}
\end{figure}So, in case of a two-participant market (say agents $i$ and $j$),
we may represent the strategy space of market participants using the
chess-like board in Figure \vref{fig:Tatonnement}. Every strategy
pair that the agents decide on can be represented as a point on this
board. For example, suppose agent $i$ initially chooses a point in
$\mathcal{D}_{i}^{4}=\mathcal{UD}_{i}^{3}\setminus\mathcal{UD}_{i}^{4}$,
while agent $j$ initially chooses a point in $\mathcal{D}_{j}^{6}=\mathcal{UD}_{j}^{5}\setminus\mathcal{UD}_{j}^{6}$.
This pair of strategies is the red cross on the left in the figure.
Suppose the responsiveness of the market aggregation mapping is such
that $i$ cannot distinguish $j$'s responses in the horizontally
hatched region, and $j$ cannot distinguish $i$'s responses in the
vertically hatched region. Then, $i$'s initial choice is dominated-wrtp
while $j$'s initial choice is not. Alternatively, using the reasoning
based narrative, we have that agent $i$ chose a belief hierarchy
in $W_{i}^{3}$ (i.e., $k=3)$ initially, then found out that $A_{-i}(\tilde{m},a_{i})$
lay inside the set $\mathcal{UD}_{j}^{4}$, which implies $A_{-i}(\tilde{m},a_{i})\subset\mathcal{UD}_{j}^{k-1}$
(with $k=3$). Thus, $i$'s status satisfies the initial condition
in (i) in Proposition \ref{prop:tatonnement and reasoning}. By the
same token, agent $j$ chose a strategy in $W_{j}^{5}$ (i.e., $k=5)$
and found out that $A_{-j}(\tilde{m},a_{j})$ was a member of $\mathcal{UD}_{i}^{3}$,
which implies $A_{-j}(\tilde{m},a_{j})\not\subset\mathcal{UD}_{i}^{k-1}$
(with $k=5$). Thus $j$'s status does not satisfy the initial condition
in Proposition \ref{prop:tatonnement and reasoning}. After realizing
that she has a tradeable arbitrage opportunity, suppose agent $i$
selects a new response in $\mathcal{D}_{i}^{7}=\mathcal{UD}_{i}^{6}\setminus\mathcal{UD}_{i}^{7}$.
The pair of strategies of the agents is now the blue cross on the
right in the figure. Now, agent $i$'s strategy is undominated-wrtp.
In the reasoning based narrative, she now uses a belief hierarchy
in $W_{i}^{6}$ (i.e., $k=6)$. $A_{-i}(\tilde{m},a_{i})$ is still
in the set $\mathcal{UD}_{j}^{4}$ from $i$'s perspective since $j$'s
response is unchanged from earlier, and therefore $A_{-i}(\tilde{m},a_{i})\supset\mathcal{UD}_{j}^{k}$
(with $k=6$). Agent $i$'s status now satisfies the subsequent condition
in (ii) in Proposition \ref{prop:tatonnement and reasoning}. $i$'s
belief hierarchy set changed from $W_{i}^{3}$ to $W_{i}^{6}$, so
$\alpha$ in this case is $3$.

The object of Proposition \ref{prop:tatonnement and reasoning} and
the example is to highlight the point that an outside anayst can build
an alternative narrative for most tatonnement adjustments in the market
using the technique of belief hierarchies and higher order reasoning,
even if market participants are not deliberately engaging in the reasoning.
This is because any tatonnement adjustment process traces out a trajectory
in the $I$-dimensional space $A_{1}\times A_{2}\times\dots\times A_{I}$,
through the grids formed by the partitions of $A_{i}$, $i\in I$
(in the figure, the squares formed by overlapping $\mathcal{D}_{i}^{k}$
and $\mathcal{D}_{j}^{k}$). 

The only tatonnement adjustments that an analyst cannot capture using
belief hierarchies are the ones that happen \emph{within} the grids.
When the grids formed by the partition of strategy spaces described
in (\ref{eq:partition D_k}), based on the belief hierarchy set $\bigcap_{k\geq0}W_{i}^{k}$,
are too coarse to capture the move from dominated-wrtp to undominated-wrtp
responses, the reasoning based narrative is no longer sufficient to
describe the tatonnement adjustments.\footnote{\label{fn:Alpha_Zero}This corresponds to the scenario $\alpha=0$
in Proposition \ref{prop:tatonnement and reasoning}, in which case
both the initial and subsequent condition in the proposition have
the participant using a belief hierarchy in the same set $W_{i}^{k}$.} In such a case, though, the analyst can move to a different belief
hierarchy set (not based on $W_{i}$) that gives finer partitions
of $A_{i}$, $i\in I$. As mentioned before, there is nothing sacrosanct
about the set $W_{i}^{k}$, and we can undertake similar analysis
with other hierarchy sets that satisfy the requisite properties in
Section \ref{subsec:Dominated-Responses}.

\section{Discussion \label{sec:Discussion}}

Understanding arbitrage or the lack of it, in markets, has been a
central preoccupation of finance, and this paper provides a new characterization
of arbitrage using belief hierarchies of Bayesian market participants.
It is shown that an arbitrage opportunity exists only when a market
participant underestimates the order of belief hierarchies actually
used for reasoning in the market. The arbitrage trade depends both
on the responsiveness of the market aggregation mapping and the order
of the belief hierarchies employed by market agents, and the more
responsive the aggregation mapping, the higher the order of reasoning
needed from agents to avoid arbitrage. In the Bayesian belief hierarchy
based approach to markets, therefore, no-arbitrage seems to go almost
hand-in-hand with higher order beliefs of participants. Can we measure
from empirical data the extent to which higher order beliefs get used
by market participants?

There seem to be multiple ways to do this using our analysis. What
we have termed responsiveness of the market aggregation mapping is
not very different, empirically, from the price impact function \textemdash{}
the degree to which asset prices move in response to a trade. In case
there is no price impact at all, no matter what trade is undertaken,
the aggregation mapping is likely to be unresponsive. As we discussed
earlier, this means market participants do not have to employ higher
order beliefs to achieve a state of no-arbitrage. On the other hand,
when the price impact of trades is high, the aggregation mapping is
likely to be more responsive and the use of higher order beliefs becomes
necessary for no-arbitrage. Therefore, the magnitude of the price
impact function can serve as a proxy for the extent to which higher
order beliefs are in use in a market, assuming the asset prices provide
no arbitrage opportunities. Of course, such an assertion comes with
a number of caveats. Price impact often has a component attributed
to asymmetric information, and since our model has no explicit role
for information asymmetry, this component would need to be carefully
decoupled. Further, the precise nature of the relationship between
price impact and responsiveness of market aggregation mapping would
need more careful consideration before taking such a methodology to
the data. With the responsiveness of the aggregation mapping, what
we are after is how well agents can distinguish among market choices
that support an observed price. The size of the trade, which is the
main ingredient for measuring traditional price impact, is only one
dimension along which market choices may be distinguished. In this
sense, price impact is a necessary condition for market aggregation
responsiveness when its value is low, and sufficient condition for
market aggregation responsiveness when its value is high. How might
one enrich the traditional measurement of price impact so that we
have a measure that is closer to a necessary\emph{ and} sufficient
condition? Empirical methodology questions like these would need to
be carefully addressed if we have to establish a credible procedure
for determining higher order beliefs from the data through price impact
functions.

Another avenue to measure the extent of higher order beliefs in markets
could be through surveys. There is a growing literature that uses
survey data to decipher belief formation for asset prices (see \citet{key-50}
for a recent survey) but the emphasis in this area has been mostly
on first order beliefs.\footnote{In the empirical macroeconomics literature, researchers have recently
started to work with higher order expectations data. Notably, \citet{key-200}
provides an analysis of survey data that estimates higher order macroeconomic
expectations of firm managers.} How might we gauge through surveys if higher order beliefs are in
use? The analysis in this paper suggests that higher the order of
beliefs over which a market participant optimizes, smaller the set
of strategies she considers feasible for other market participants.
So, for instance, periodically surveying hedge fund managers about
how they think Reddit-based retail traders might trade, or surveying
retail traders about how institutional investors might trade, and
so on, can provide a way to decipher the time variation of higher
order beliefs in the market. If a category of traders attributes a
wide variety of strategies to another category, it is likely that
the order of beliefs at play is low. On the other hand, if the attributed
strategies belong to a narrow set, it is quite likely that higher
order beliefs are being used. Another possibility could be to include
tatonnement-styled questions in the surveys that query traders on
their strategies in a series of hypothetical scenarios. Responses
to such scenario-based questions could then be used to gauge the degree
to which traders stand ready to engage in the back and forth needed
for a market adjustment process. Needless to say, these are all very
rough outlines, and a lot more work needs to be undertaken to establish
such methodologies in practice.

From an applied perspective, one of the main takeaways from the paper
is that even the most basic of concepts in finance \textemdash{} the
notion of arbitrage \textemdash{} is inextricably linked to higher
order beliefs. There is increasing recognition in the finance community
that asset prices may move despite the lack of commensurate movement
in fundamentals, and a key to disentangling such puzzling market movements
likely lies in higher order beliefs. Exploring the role of Bayesian
belief hierarchies in more detail in the context of financial markets
should thus be a promising agenda of research.

\appendix

\section*{Appendix: Omitted Proofs\label{sec:Appendix}}
\begin{proof}
\textsf{\small{}(Propostion \ref{prop:Amenable-Arb-def}). }{\small\par}

\textsf{\emph{\small{}Part-1}}\textsf{\small{}: }\textsf{\emph{\small{}There
is a tradeable arbitrage opportunity in the market implies the condition
in the Proposition holds.}}{\small\par}

\textsf{\small{}From Definition \ref{def:With-Price-impact-Arb Opportunity},
there is a tradeable arbitrage opportunity if and only if agent $i$
can trade a portfolio $\theta_{i}\in\mathbb{R}^{d}$ with the property
\begin{equation}
\theta_{i}\cdot q_{\theta_{i}}\leq0,\,\,\text{but}\,\,\theta_{i}.\tilde{x}\geq0\,\mathbf{P}\text{-a.s.}\,\,\text{and}\,\,\mathbf{P}[\theta_{i}.\tilde{x}>0]>0.\label{eq:Lem arb-from-def}
\end{equation}
Let us use the superscript }\textsf{\emph{\small{}$rf$}}\textsf{\small{}
to label the risk-free asset; i.e., the risk-free asset's ex-ante
price is $q_{\theta_{i}}^{rf}$, it's payout is $x^{rf}$, and the
portfolio weight on risk-free is $\theta^{rf}$. Let us use $q_{\theta_{i}}^{-rf},$
$\tilde{x}^{-rf}$ and $\theta^{-rf}$ to designate the corresponding
variables for the set of risky assets. Since $\theta_{i}\cdot q_{\theta_{i}}\leq0$,
and $\int_{S}\tilde{m}\mathrm{d}\mathbf{P}>0$ is positive for all
$\tilde{m}$ under consideration, we have $0\geq\frac{\theta_{i}\cdot q_{\theta_{i}}}{\int_{S}\tilde{m}\mathrm{d}\mathbf{P}}=\frac{\theta_{i}^{rf}\cdot q_{\theta_{i}}^{rf}}{\int_{S}\tilde{m}\mathrm{d}\mathbf{P}}+\frac{\theta_{i}^{-rf}\cdot q_{\theta_{i}}^{-rf}}{\int_{S}\tilde{m}\mathrm{d}\mathbf{P}}$.
This implies 
\begin{equation}
\theta_{i}^{-rf}\cdot\tilde{x}^{-rf}-\frac{\theta_{i}^{-rf}\cdot q_{\theta_{i}}^{-rf}}{\int_{S}\tilde{m}\mathrm{d}\mathbf{P}}\geq\theta_{i}^{-rf}\cdot\tilde{x}^{-rf}+\frac{\theta_{i}^{rf}\cdot q_{\theta_{i}}^{rf}}{\int_{S}\tilde{m}\mathrm{d}\mathbf{P}}=\theta_{i}.\tilde{x}
\end{equation}
for all $\tilde{m}$. Since $\theta_{i}.\tilde{x}$ is $\mathbb{\mathbf{P}\text{-a.s.}}$
non-negative and strictly positive with non-vanishing probability
from (\ref{eq:Lem arb-from-def}), the same must be true of $\theta_{i}^{-rf}\cdot\tilde{x}^{-rf}-\frac{\theta_{i}^{-rf}\cdot q_{\theta_{i}}^{-rf}}{\int_{S}\tilde{m}\mathrm{d}\mathbf{P}}$.
As for the risk-free asset, we have $\theta_{i}^{rf}\cdot x^{rf}-\frac{\theta_{i}^{rf}\cdot q_{\theta_{i}}^{rf}}{\int_{S}\tilde{m}\mathrm{d}\mathbf{P}}=0$.
Thus, from (\ref{eq:net gain vector}), we have that the gains from
trading the portfolio $\theta_{i}$ satisfy}{\small\par}

\textsf{\small{}
\begin{equation}
\tilde{g}(\theta_{i},a_{-i})\geq0\,\mathbf{P}\text{-a.s.}\,\,\text{and}\,\,\mathbf{P}[\tilde{g}(\theta_{i},a_{-i})>0]>0.\label{eq:arb again}
\end{equation}
}{\small\par}

\textsf{\small{}We are given that $a_{i}$ is agent $i$'s original
strategy. Define her new strategy $a_{i}^{*}$ as the composition
operation applied to her original strategy and the arbitrage trade.
That is to say, the original strategy followed by the arbitrage trade
is her new strategy, so that 
\begin{equation}
\tilde{g}(a_{i},a_{-i})+\tilde{g}(\theta_{i},a_{-i})=\tilde{g}(a_{i}^{*},a_{-i}).
\end{equation}
Since $\tilde{g}(\theta_{i},a_{-i})$ satisfies the condition in (\ref{eq:arb again}),
we get 
\begin{equation}
\tilde{g}(a_{i}^{*},a_{-i})-\tilde{g}(a_{i},a_{-i})\geq0\,\mathbf{P}\text{-a.s.}\,\,\text{and}\,\,\mathbf{P}[\tilde{g}(a_{i}^{*},a_{-i})-\tilde{g}(a_{i},a_{-i})>0]>0.
\end{equation}
Using the equivalence in (\ref{eq:util and action relation}) we therefore
obtain the requisite relation among the utilities listed in (\ref{eq:arbitrage utilities}).}{\small\par}

\textsf{\emph{\small{}Part-}}\textsf{\small{}2: }\textsf{\emph{\small{}The
condition in the Proposition holds implies there is a tradeable arbitrage
opportunity in the market.}}{\small\par}

\textsf{\small{}In this case we are given that $i$ may change her
strategy from $a_{i}$ to $a_{i}^{*}$ to obtain
\begin{equation}
\tilde{U}_{i}(a_{i}^{*},a_{-i})\geq\tilde{U}_{i}(a_{i},a_{-i})\,\mathbf{P}\text{-a.s.}\,\,\text{and}\,\,\mathbf{P}[\tilde{U}_{i}(a_{i}^{*},a_{-i})>\tilde{U}_{i}(a_{i},a_{-i})]>0.\label{eq:Integral rel betw m's}
\end{equation}
 Using the equivalence in (\ref{eq:util and action relation}), we
get
\begin{equation}
\tilde{g}(a_{i}^{*},a_{-i})-\tilde{g}(a_{i},a_{-i})\geq0\,\mathbf{P}\text{-a.s.}\,\,\text{and}\,\,\mathbf{P}[\tilde{g}(a_{i}^{*},a_{-i})-\tilde{g}(a_{i},a_{-i})>0]>0.\label{eq:gain relation between actions}
\end{equation}
}{\small\par}

\textsf{\small{}Construct the following portfolio: (i) Go long the
portfolio $a_{i}^{*}$ by investing $a_{i}^{*}\cdot q_{a_{i}^{*}}$
(ii) Go short the portfolio $a_{i}$ by investing $-a_{i}\cdot q_{a_{i}}$,
(iii) Fund the investments by trading $-a_{i}^{*}\cdot q_{a_{i}^{*}}$
in the risk-free asset initially, followed by another trade of $+a_{i}\cdot q_{a_{i}}$
in the risk free asset. Here, $q_{a_{i}^{*}}=\int_{S}\tilde{f}(a_{i}^{*},a_{-i})\tilde{x}\mathrm{d}\mathbf{P}$,
$q_{a_{i}}=\int_{S}\tilde{f}(a_{i},a_{-i})\tilde{x}\mathrm{d}\mathbf{P}$. }{\small\par}

\textsf{\small{}We claim that this is a tradeable arbitrage portfolio.
Indeed, the ex-ante investment in constructing the portfolio is zero.
The stochastic payout from (i) and (ii) is 
\begin{equation}
a_{i}^{*}\tilde{x}-a_{i}\tilde{x},\label{eq:S_payout from risky}
\end{equation}
and from (iii) is
\begin{equation}
\frac{-a_{i}^{*}\cdot q_{a_{i}^{*}}}{\int_{S}\tilde{m}_{-a_{i}^{*}\cdot q_{a_{i}^{*}}}^{rf}\mathrm{d}\mathbf{P}}+\frac{a_{i}\cdot q_{a_{i}}}{\int_{S}\tilde{m}_{+a_{i}\cdot q_{a_{i}}}^{rf}\mathrm{d}\mathbf{P}}.\label{eq:S_payout from r_free}
\end{equation}
Summing (\ref{eq:S_payout from risky}) and (\ref{eq:S_payout from r_free}),
we get that the total stochastic payout is $\tilde{g}(a_{i}^{*},a_{-i})-\tilde{g}(a_{i},a_{-i})$,
which is $\mathbf{P}$-a.s non-negative and strictly positive with
nonvanishing probability, by (\ref{eq:gain relation between actions}).
Therefore, from (\ref{eq:Arb defn-3}), this is an arbitrage trade.}{\small\par}
\end{proof}
\medskip

\begin{proof}
\textsf{\small{}(Proposition \ref{prop:embedding belief-hierarchy}).
The result follows from the Axiom of specification in Set theory which
states that:} \textsf{\small{}To every set $\mathscr{A}$ and to every
condition $S(x)$ there corresponds a set $\mathscr{B}\subseteq\mathscr{A}$
whose elements are exactly those elements $x$ of $\mathscr{A}$ for
which $S(x)$ holds.}\footnote{\textsf{\footnotesize{}See \citet{key-25}, Chapter 2.}}\textsf{\small{} }{\small\par}

\textsf{\small{}From equation (\ref{eq:set W_i^2}), $W_{i}^{2}$
imposes the additional condition $\marg_{A_{i}}{[b_{j}^{1}]}[a_{i}\in D_{i}^{1}]=0$
on members of the set $B_{j}$. Therefore, by the axiom of specification
$B_{j}\supseteq b_{j}(W_{i}^{2})$. Similarly, from equation (\ref{eq:Belief hierarchy set W_i^k}),
$W_{i}^{k}$ imposes the additional condition $\marg_{A_{i}}{[b_{j}^{k-1}]}[a_{i}\in D_{i}^{k-1}]=0$
on members of $b_{j}(W_{i}^{k-1})$. So, by the axiom of specification
$b_{j}(W_{i}^{k-1})\supseteq b_{j}(W_{i}^{k})$. Thus, we have the
relation in the statement of the proposition.}{\small\par}
\end{proof}
\medskip
\begin{proof}
\textsf{\small{}(Proposition \ref{prop:Responsive actions}). One
way to prove this proposition is to follow the strategy in the proof
of Proposition \ref{prop:embedding belief-hierarchy} and use the
Axiom of specification. }{\small\par}

\textsf{\small{}A more direct way is to observe that equation (\ref{eq:Belief hierarchy set V_i^k})
imposes the restriction that any belief hierarchy $b_{i}$ in $V_{i}^{k}$
and $W_{i}^{k}$ satisfies $\marg_{A_{j}}{[\phi_{i}(b_{i})]}[a_{j}\in\mathcal{D}_{j}^{k-1}]=0$.
Next, since $b_{i}\in W_{i}^{k}$ only if $b_{i}\in W_{i}^{k-1}$,
we get, additionally, $\marg_{A_{j}}{[\phi_{i}(b_{i})]}[a_{j}\in\mathcal{D}_{j}^{k-2}]=0$.
Iterating this argument repeatedly, we obtain 
\begin{equation}
\marg_{A_{j}}{[\phi_{i}(b_{i})]}[a_{j}\in\mathcal{D}_{j}^{k-n}]=0\text{ for all }k-1\leq n\leq1.
\end{equation}
Therefore, $\supp{\marg_{A_{j}}{[\phi_{i}(b_{i})]}}=A_{j}\setminus\bigcup_{n=1}^{k-1}\mathcal{D}_{j}^{n}=\mathcal{UD}_{j}^{k-1}$,
which implies $A_{j}(W_{i}^{k})=\mathcal{UD}_{j}^{k-1}$. Finally,
by the definition of undominated response sets in equations (\ref{eq:undominated W^k})
and (\ref{eq:undominated W_i^1}), we obtain 
\begin{equation}
A_{j}\supseteq\mathcal{UD}_{j}^{1}\supseteq\dots\supseteq\mathcal{UD}_{j}^{k-1}\supseteq\mathcal{UD}_{j}^{k}\supseteq\dots
\end{equation}
Thus, we have the relation in the statement of the proposition.}{\small\par}
\end{proof}
\medskip
\begin{proof}
\textsf{\small{}(Proposition \ref{prop:f is one-one}). We have shown
sufficiency in part-1b in the proof of Theorem \ref{thm:No-arbitrage necc =000026 suff}.
To establish necessity, we need to rule out the possibility that $k$
is finite. When $k$ is finite, there are two scenarios that may result.
In both the scenarios, we show below, there is an arbitrage opportunity.
This implies $k$ cannot be finite, i.e. all market participants choose
a response in $\bigcap_{k\geq0}\mathcal{UD}_{i}^{k}$.}{\small\par}

\textsf{\small{}Scenario 1. All market participants use the same finite
order $k$ for their belief hierarchy and choose responses in the
set $\mathcal{UD}_{i}^{k}\setminus\mathcal{UD}_{i}^{k+1}$, $i\in I$:}{\small\par}

\textsf{\small{}In this case, every market participant finds a tradeable
arbitrage opportunity. This is because from condition (\ref{eq:undominated W^k})
in Definition \ref{def:SD level k}, we have $\mathcal{UD}_{i}^{k}\setminus\mathcal{UD}_{i}^{k+1}=\mathcal{D}_{i}^{k+1}$,
and the definition of $\mathcal{D}_{i}^{k+1}$ implies that agent
$i$, $i\in I$, may change her strategy from $a_{i}$ to $a_{i}^{*}$
to obtain
\begin{equation}
\tilde{U}_{i}(a_{i}^{*},a_{-i})\geq\tilde{U}_{i}(a_{i},a_{-i})\,\mathbf{P}\text{-a.s.}\text{ and}\,\,\mathbf{P}[\tilde{U}_{i}(a_{i}^{*},a_{-i})>\tilde{U}_{i}(a_{i},a_{-i})]>0,\label{eq:dominated for BH to arb-again}
\end{equation}
for any $a_{-i}\in\prod_{j\neq i}\mathcal{UD}_{j}^{k}$. Since $\tilde{f}$
is one-to-one, the agents can accurately discern the strategies used
by other market participants. So $A_{-i}(\tilde{m},a_{i})=\prod_{j\neq i}\mathcal{UD}_{j}^{k}$.
Theorem \ref{thm:From BH to Arb}, then, implies that each agent $i\in I$
has a tradeable arbitrage opportunity.}{\small\par}

\textsf{\small{}Scenario 2. Different market participants use different
finite orders $k_{i}$ for their belief hierarchy and choose responses
in the set $\mathcal{UD}_{i}^{k_{i}}\setminus\mathcal{UD}_{i}^{k_{i}+1}$,
$i\in I$:}{\small\par}

\textsf{\small{}In this case, the participant that chooses $k_{min}=\min\{k_{i}:i\in I\}$
\textemdash{} call her $i_{min}$ \textemdash{} has a tradeable arbitrage
opportunity. This is because, similar to Case 1 above, $\mathcal{UD}_{i}^{k_{min}}\setminus\mathcal{UD}_{i}^{k_{min}+1}=\mathcal{D}_{i}^{k_{min}+1}$,
and the definition of $\mathcal{D}_{i}^{k_{min}+1}$ implies that
agent $i_{min}$, may change her strategy from $a_{i}$ to $a_{i}^{*}$
to obtain
\begin{equation}
\tilde{U}_{i}(a_{i}^{*},a_{-i})\geq\tilde{U}_{i}(a_{i},a_{-i})\,\mathbf{P}\text{-a.s.}\text{ and}\,\,\mathbf{P}[\tilde{U}_{i}(a_{i}^{*},a_{-i})>\tilde{U}_{i}(a_{i},a_{-i})]>0,\label{eq:dominated for BH to arb-again-1}
\end{equation}
for any $a_{-i_{min}}\in\prod_{j\neq i}\mathcal{UD}_{j}^{k_{min}}$.
Since $\tilde{f}$ is one-to-one, the agents can accurately discern
the strategies used by other market participants. So, $A_{-i_{min}}(\tilde{m},a_{i})\subseteq\prod_{j\neq i}\mathcal{UD}_{j}^{k_{min}}$.
Theorem \ref{thm:From BH to Arb}, then, implies that agent $i_{min}$
has a tradeable arbitrage opportunity.}{\small\par}

\textsf{\small{}Scenarios 1 and 2 were described with all market participants
using finite orders, but we could make analogous arguments for cases
where only a subset of the market participants were using finite orders.
Thus, $k$ cannot be finite for any market participant, given there
is no tradeable arbitrage opportunity.}{\small\par}
\end{proof}
\medskip
\begin{proof}
\textsf{\small{}(Proposition \ref{prop:f responsiveness}). From the
definition of responsiveness, $A_{-i}^{\tilde{f_{2}}}(\tilde{m},a_{i})\subseteq A_{-i}^{\tilde{f_{1}}}(\tilde{m},a_{i})$
when $\tilde{f}_{2}$ is more responsive than $\tilde{f}_{1}$, for
$i\in I$. So, if $A_{-i}^{\tilde{f_{2}}}(\tilde{m},a_{i})\supset\prod_{j\neq i}\mathcal{UD}_{j}^{k_{i2}}$,
then we must also have $A_{-i}^{\tilde{f_{1}}}(\tilde{m},a_{i})\supset\prod_{j\neq i}\mathcal{UD}_{j}^{k_{i2}}$.
As noted before, from Definition \ref{def:SD level k}, we have that
undominated response sets satisfy the property $\mathcal{UD}{}_{j}^{m1}\supseteq\mathcal{UD}{}_{j}^{m2}$
when $m1\leq m2$, i.e. the undominated response sets form a nested
sequence as one increases the order. Therefore, if $k_{i1}$ is the
}\textsf{\emph{\small{}minimum}}\textsf{\small{} order for which $A_{-i}^{\tilde{f_{1}}}(\tilde{m},a_{i})\supseteq\prod_{j\neq i}\mathcal{UD}_{j}^{k_{i1}}$,
we must have $k_{i1}\leq k_{i2}$.}{\small\par}

\textsf{\small{}From Theorem \ref{thm:No-arbitrage necc =000026 suff}-(ii),
under no-arbitrage, the order of the undominated response under $\tilde{f}_{1}$
needs to be at least $k_{i1}$ given $A_{-i}^{\tilde{f_{1}}}(\tilde{m},a_{i})\supset\prod_{j\neq i}\mathcal{UD}_{j}^{k_{i1}}$,
while under $\tilde{f}_{2}$ it needs to be at least $k_{i2}$ given
$A_{-i}^{\tilde{f_{2}}}(\tilde{m},a_{i})\supset\prod_{j\neq i}\mathcal{UD}_{j}^{k_{i2}-1}$,
$i\in I$. Thus, the minimum order of the undominated response set
from which the market participant selects, given no-arbitrage, increases
weakly under $\tilde{f}_{2}$.}{\small\par}
\end{proof}
\medskip
\begin{proof}
\textsf{\small{}(Proposition \ref{prop:tatonnement and reasoning}).
Since we are focusing exclusively on agent $i$'s strategies here,
we will assume without loss of generality that no other market participant
has a tradeable arbitrage opportunity.}\footnote{\textsf{\footnotesize{}It is easy to see that each proposition above
on the existence of a tradeable arbitrage opportunity }\textsf{\emph{\footnotesize{}in
the market}}\textsf{\footnotesize{} can be restated as a proposition
about the existence of a tradeable arbitrage opportunity }\textsf{\emph{\footnotesize{}for
agent}}\textsf{\footnotesize{} $i$.}}{\small\par}

$a_{i}\in D_{i}^{wrtp}$ \textsf{\small{}implies, from Corollary \ref{cor:Corollary Amenable Arb},
that agent $i$ has a tradeable arbitrage opportunity. From Theorem
\ref{thm:From Arb to BH}, we then get the condition that agent $i$
uses a belief hierarchy in $W_{i}^{k}$ for the eductive sequence
and finds 
\begin{equation}
A_{-i}(\tilde{m},a_{i})\subset\prod_{j\neq i}\mathcal{UD}_{j}^{k-1}.\label{eq:first step rel A_-i and UD}
\end{equation}
 This gives the correspondence of $a_{i}\in D_{i}^{wrtp}$ with (i)
in the statement of the Proposition.}{\small\par}

\textsf{\small{}Next, $a_{i}\in UD_{i}^{wrtp}$ implies, from Corollary
\ref{cor:Corollary tatonnement}, that agent $i$ does not have a
tradeable arbitrage opportunity. From Theorem \ref{thm:No-arbitrage necc =000026 suff},
we then get the condition that agent $i$ either chooses a response
in the set $\mathcal{UD}_{i}^{k_{i}}\setminus\mathcal{UD}_{i}^{k_{i}+1}$
and finds $A_{-i}(\tilde{m},a_{i})\supset\prod_{j\neq i}\mathcal{UD}_{j}^{k_{i}}$,
or chooses a response in the set $\bigcap_{n\geq0}\mathcal{UD}_{i}^{n}$.
We show below that the first scenario is equivalent to the agent using
a belief hierarchy in $W_{i}^{k+\alpha}$ for the eductive sequence
and finding $A_{-i}(\tilde{m},a_{i})\supset\prod_{j\neq i}\mathcal{UD}_{j}^{k+\alpha}$,
and the second scenario is equivalent to the agent using a belief
hierarchy in $\bigcap_{\beta>\alpha}W_{i}^{k+\beta}$. This gives
the correspondence of }$a_{i}\in UD_{i}^{wrtp}$\textsf{\small{} with
(ii) in the statement of the Proposition. Let us look at each scenario
in turn:}{\small\par}

\textsf{\small{}Scenario 1. Agent $i$ chooses a response in the set
$\mathcal{UD}_{i}^{k_{i}}\setminus\mathcal{UD}_{i}^{k_{i}+1}$ and
finds $A_{-i}(\tilde{m},a_{i})\supset\prod_{j\neq i}\mathcal{UD}_{j}^{k_{i}}$:}{\small\par}

\textsf{\small{}If the agent chooses a response in the set $\mathcal{UD}_{i}^{k_{i}}$
she is using a belief hierarchy in $W_{i}^{k_{i}}$ for the purpose
of the eductive sequence. Thus, to establish the requisite equivalence,
we need to show $k_{i}=k+\alpha$, where the value of $k$ is fixed
by the initial belief hierarchy $W_{i}^{k}$ in (i) of the statement
of the Proposition. As noted before, from Definition \ref{def:SD level k}
we have that undominated response sets satisfy the property $\mathcal{UD}{}_{j}^{m1}\supseteq\mathcal{UD}{}_{j}^{m2}$
when $m1\leq m2$, i.e. the undominated response sets form a nested
sequence as one climbs up the orders. This implies that if $\mathcal{UD}{}_{j}^{m1}\supset\mathcal{UD}{}_{j}^{m2}$,
then $m1<m2$. Therefore, from $A_{-i}(\tilde{m},a_{i})\subset\prod_{j\neq i}\mathcal{UD}_{j}^{k-1}$
in (\ref{eq:first step rel A_-i and UD}), and $A_{-i}(\tilde{m},a_{i})\supset\prod_{j\neq i}\mathcal{UD}_{j}^{k_{i}}$
in the present scenario, we get that
\begin{equation}
k_{i}>k-1.
\end{equation}
This implies $k_{i}=k+\alpha$, for some integer $\alpha\geq0$, as
needed. }{\small\par}

\textsf{\small{}Scenario 2. Agent $i$ chooses a response in the set
$\bigcap_{n\geq0}\mathcal{UD}_{i}^{n}$: }{\small\par}

\textsf{\small{}An undominated $k^{th}$ order response $\mathcal{UD}_{i}^{k}$
is defined with respect to the belief hierarchy $W_{i}^{k}$ (Definition
\ref{def:SD level k}), so in this case, given the response is in
the set $\bigcap_{n\geq0}\mathcal{UD}_{i}^{n}$, $i$ employs a belief
hierarchy in $\bigcap_{n\geq0}W_{i}^{n}$. Next, we have that the
set $\bigcap_{n\geq0}W_{i}^{n}=\bigcap_{\beta>\alpha}W_{i}^{k+\beta}$,
for any given $\alpha\geq0$, because the belief hierarchy sets $W_{i}^{n}$
are nested, by definition (see equations (\ref{eq:set V_i^2}) \textendash{}
(\ref{eq:restriction_general hierarchy to action})). That is to say,
$W_{i}^{0}\supseteq W_{i}^{1}\supseteq W_{i}^{2}\supseteq\dots$ Therefore,
the agent using a belief hierarchy in} $\bigcap_{\beta>\alpha}W_{i}^{k+\beta}$\textsf{\small{}
is equivalent to her choosing a response in the set $\bigcap_{n\geq0}\mathcal{UD}_{i}^{n}$.}{\small\par}
\end{proof}
\medskip

\medskip
\begin{thebibliography}{Banerjee, Kaniel and Kremer(2009)}
\bibitem[Adam and Marcet(2011)]{key-210}Adam, K. and Marcet, A. (2011).
Internal rationality, imperfect market knowledge and asset prices.
Journal of Economic Theory, 146, 1224-1252.

\bibitem[Adam and Nagel(2022)]{key-50}Adam, K. and Nagel, S. (2022).
Expectations Data in Asset Pricing. Forthcoming in the Handbook of
Economic Expectations.

\bibitem[Allen, Morris and Shin (2006)]{key-34}Allen, F., Morris,
S. and Shin, H. S. (2006). Beauty Contests, Bubbles and Iterated Expectations
in Asset Markets. Review of Financial Studies 19, 719\textendash 52.

\bibitem[Armbruster \& Böge(1979)]{key-24}Armbruster, W. and Böge,
W. (1979). Bayesian game theory. Game Theory and Related Topics. North-Holland,
Amsterdam. 

\bibitem[Back(2017)]{key-6}Back, K. E. (2017). Asset Pricing and
Portfolio Choice Theory (2nd edition). Oxford University Press, New
York.

\bibitem[Banerjee, Kaniel and Kremer(2009)]{key-36}Banerjee, S, Kaniel,
R. and Kremer I. (2009). Price drift as an outcome of differences
in higher order beliefs. Review of Financial Studies, 22, 3707\textendash 3734.

\bibitem[Ben-Porath and Heifetz(2011)]{key-44}Ben-Porath, E. and
Heifetz, A. (2011). Common knowledge of rationality and market clearing
in economies with asymmetric information. Journal of Economic Theory,
146, 2608\textendash 2626.

\bibitem[Bernheim(1984)]{key-5}Bernheim, B. D. (1984). Rationalizable
Strategic Behavior. Econometrica, 52(4), 1007-1028.

\bibitem[Biais and Bossaerts(1993)]{key-37}Biais, B. and Bossaerts,
P. (1993). Asset prices and trading volume in a beauty contest. Review
of Economic Studies, 65, 307\textendash 340.

\bibitem[Binmore(1987)]{key-28}Binmore, K. (1987). Modeling Rational
Players: Part 1. Economics and Philosophy, 1987(3), 179-214.

\bibitem[Boge and Eisele(1979)]{key-19}Boge, W. and Eisele Th. (1979).
On solutions of Bayesian games. International Journey of Game Theory,
8(4), 193\textendash 215.

\bibitem[Brandenburger and Dekel(1993)]{key-9}Brandenburger, A. and
Dekel, E. (1993). Hierarchies of Beliefs and Common Knowledge. Journal
of Economic Theory, 59(1), 189-198.

\bibitem[Brandenburger and Dekel(1987)]{key-21} Brandenburger, A.
and Dekel, E. (1987). Rationalizability and Correlated Equilibria.
Econometrica, 55(6), 1391-1402.

\bibitem[Coibion et al.(2021)]{key-200} Coibion, O., Gorodnichenko,
Y., Kumar S. and Ryngaert, J. (2021). Do You Know that I Know that
You Know. . . ? Higher-Order Beliefs in Survey Data. Quarterly Journal
of Economics, 136(3), 1387\textendash 1446

\bibitem[Dekel and Siniscalchi(2015)]{key-11} Dekel, E. and Siniscalchi,
M. (2015). Epistemic Game Theory, Chapter 12 in Handbook of Game Theory
with Economic Applications, Vol. 4, Elselvier, 619-702. 

\bibitem[Desgranges(2014)]{key-41} Desgranges, G. (2014). Coordination
of Expectations: The Eductive Stability Viewpoint. Annual Review of
Economics, 6, 273\textendash 98.

\bibitem[Dutta and Morris(1997)]{key-42} Dutta, J. and Morris, S.
The revelation of information and self-fulfilling beliefs. Journal
of Economic Theory, 73, 231\textendash 244.

\bibitem[Dybvig and Ross(1987)]{key-18}Dybvig, P.H. and Ross, S.A.
(1987). Arbitrage. In: J. Eatwell, M. Milgate and P. Newman, eds.,
The New Palgrave, A Dictionary of Economics, (London: The MacMillan
Press, Ltd.). 1, 100-106.

\bibitem[Foucault et al.(2013)]{key-202} Foucault, T., Pagano and
Röell, A. (2013). Market Liquidity: Theory, Evidence, and Policy.
Oxford University Press, New York, NY, USA.

\bibitem[Gromb and Vayanos(2010)]{key-26}Gromb, D. and Vayanos D.
(2010). Limits of Arbitrage. Annual Review of Financial Economics,
2(1), 251-275.

\bibitem[Guesneries(1992)]{key-40}Guesnerie R. (1992). An exploration
on the eductive justifications of the rational-expectations hypothesis.
American Economic Review, 82,1254\textendash 78.

\bibitem[Guesnerie(2002)]{key-38}Guesnerie, R. (2002). Anchoring
economic predictions in common knowledge, Econometrica, 70, 439\textendash 480.

\bibitem[Guesnerie(2005)]{key-39}Guesnerie, R. (2005). Assessing
Rational Expectations 2: ``Eductive'' Stability in Economics. MIT
Press, Cambridge, MA, USA.

\bibitem[Halmos(1960)]{key-25}Halmos, P. (1960). Naive Set Theory.
D. Van Nostrand Company, Princeton, NJ, USA.

\bibitem[Harrison and Kreps(1979)]{key-16}Harrison, J. M. \& Kreps
D. M. (1979). Martingales and arbitrage in multiperiod securities
markets. Journal of Economic Theory, 20(3), 381-408.

\bibitem[Harsányi (1967)]{key-1}Harsányi, J. C. (1967-68). Games
with incomplete information played by \textquoteleft Bayesian\textquoteright{}
players, Parts I \textendash III. Management Science, 14, 159\textendash 182,
320-334, 486-502.

\bibitem[Heifetz(1993)]{key-27}Heifetz, A. (1993). The bayesian formulation
of incomplete information \textemdash{} The non-compact case. International
Journal of Game Theory, 21, 329\textendash 338. 

\bibitem[Keynes(1936)]{key-29}Keynes, J. M. (1936). The General Theory
of Employment, Interest and Money. London: Macmillan.

\bibitem[MacAllister(1990)]{key-45}MacAllister, P. (1990). Rational
behavior and rational expectations, Journal of Economic Theory, 52,
332\textendash 363.

\bibitem[Makarov and Rytchkov(2012)]{key-35}Makarov, I. and Rytchkov,
O. (2012). Forecasting the Forecasts of Others: Implications for Asset
Pricing. Journal of Economic Theory, 147, 941\textendash 966.

\bibitem[Mertens and Zamir(1985)]{key-10} Mertens, J.-F. and Zamir,
S. (1985). Formulation of Bayesian analysis for games with incomplete
information. International Journal of Game Theory, 14(1), 1\textendash 29. 

\bibitem[Morris(1995)]{key-43}Morris, S. (1995). Justifying rational
expectations, CARESS working paper.

\bibitem[Nagel(1995)]{key-49} Nagel, R. (1995). Unraveling in Guessing
Games: An Experimental Study. The American Economic Review, 85(5),
1313-1326. 

\bibitem[Pearce(1984)]{key-3}Pearce, D. G. (1984). Rationalizable
Strategic Behavior and the Problem of Perfection. Econometrica, 52(4),
1029-1050.

\bibitem[Phelps(1983)]{key-32}Phelps, E. (1983). The Trouble with
\textquoteleft \textquoteleft Rational Expectations\textquoteright \textquoteright{}
and the Problem of Inflation Stabilization. In R. Frydman and E. S.
Phelps (eds.), Individual Forecasting and Aggregate Outcomes, Cambridge
University Press, New York.

\bibitem[Rogers and Williams(2000)]{key-8} Rogers, L. C. G. and Williams,
D. (2000). Diffusions, Markov Processes and Martingales, Volume 1:
Foundations (2nd edition). Cambridge University Press, Cambridge,
UK.

\bibitem[Ross(1976)]{key-46} Ross, S. -A. (1976). The Arbitrage Theory
of Capital Asset Pricing. Journal of Economic Theory 13, 341\textendash 60.

\bibitem[Ross(1978)]{key-17}Ross, S. -A. (1978). A Simple Approach
to the Valuation of Risky Streams. Journal of Business, 51(3), 453-475.

\bibitem[Sargent(1991)]{key-33}Sargent, T. -J. (1991). Equilibrium
with Signal Extraction from Endogenous Variables. Journal of Economic
Dynamics and Control, 15, 245\textendash 273.

\bibitem[Shleifer and Vishny(1997)]{key-15}Shleifer, A. and Vishny,
R. (1997). The Limits of Arbitrage. Journal of Finance, 52 (1), 35-55.

\bibitem[Stahl and Wilson(1994)]{key-47}Stahl, D. and Wilson, P.
(1994). Experimental Evidence on Players\textquoteright{} Models of
Other Players. Journal of Economic Behavior and Organization, 25,
309\textendash 327.

\bibitem[Stahl and Wilson(1995)]{key-48}Stahl, D. and Wilson, P.
(1995). On Players\textasciiacute{} Models of Other Players: Theory
and Experimental Evidence. Games and Economic Behavior Volume, 10(1),
218-254.

\bibitem[Tan and Werlang(1988)]{key-20}Tan, T.C.C. and da Costa Werlang,
S.R. (1988). The Bayesian foundations of solution concepts of games.
Journal of Economic Theory, 45 (2), 370\textendash 391. 

\bibitem[Townsend(1978)]{key-31}Townsend, R. (1978). Market Anticipations,
Rational Expectations and Bayesian Analysis, International Economic
Review, 19, 481\textendash 94.

\bibitem[Townsend(1983)]{key-30}Townsend, R. (1983). Forecasting
the Forecasts of Others. Journal of Political Economy, 91, 546\textendash 88.
\end{thebibliography}
\end{document}